\documentclass[a4paper,11pt]{article}

\usepackage{amsmath,amssymb,mathtools,amsthm,float,setspace}
\usepackage[shortlabels]{enumitem}
\usepackage[colorlinks=true,allcolors=blue]{hyperref}
\usepackage[margin=1.72cm]{geometry}
\usepackage[breakable]{tcolorbox}
\usepackage{tgpagella,eulervm}
\newcommand{\eh}[1]{\fbox{#1}}
\definecolor{lightgray}{rgb}{0.9,0.9,0.9}
\definecolor{lightblue}{RGB}{209,228,245}
\definecolor{deterministic}{RGB}{0, 47, 108}
\definecolor{stochastic}{RGB}{216, 58, 86}
\usepackage[font=small]{caption}
\usepackage{doi}
\usepackage[sort&compress,numbers]{natbib}
\newcommand{\sgn}{\operatorname{sgn}}

\usepackage{tikz}
\newcommand{\tikznode}[2]{%
\ifmmode%
\tikz[remember picture,baseline=(#1.base),inner sep=0pt] \node (#1) {$#2$};%
\else
\tikz[remember picture,baseline=(#1.base),inner sep=0pt] \node (#1) {#2};%
\fi}

\newtheorem{definition}{Definition}
\newtheorem{proposition}{Proposition}
\newtheorem{remark}{Remark}

\newtheorem{algorithm}{Algorithm}

\usepackage{abstract}

\title{Properties of the matrix functions arising in exponential integrators with applications to stochastic simulations of PDEs}
\date{\medskip}

\author{Elliot J. Carr\footnote{Corresponding Author (\href{elliot.carr@qut.edu.au}{elliot.carr@qut.edu.au})}\\ \small School of Mathematical Sciences, Queensland University of Technology (QUT), Brisbane, Australia}

\newcommand{\algbox}[1]{\small\vspace{0.0cm}\fcolorbox{lightblue}{lightblue}{\parbox{0.99\textwidth}{\color{black}\vspace{-0.2cm}#1 \vspace{-0.2cm}}}\vspace{0.2cm}}

\begin{document}
\maketitle
\vspace*{-1.0cm}
\begin{abstract}
The matrix exponential and the closely related matrix phi-functions play a fundamental role in the solution of first-order systems of ordinary differential equations (ODEs). In particular, they appear in exact solutions of certain linear ODE systems, and in exponential integrators, a class of explicit time discretisation methods for computing approximate solutions of stiff semi-linear ODE systems. Fundamental properties of the matrix exponential include that it is nonnegative if the matrix is essentially nonnegative and stochastic if the matrix is a transition-rate matrix. In this paper, we study related properties for the matrix phi-functions. Using these properties, we then provide insights into deterministic solutions of linear and semi-linear ODE systems and outline how such solutions can be used to generate stochastic simulations to quantify variability in model predictions. The paper concludes with some illustrative examples exploring the theory for some ODE systems arising from spatial discretisation of PDEs of diffusion, advection-diffusion, Fisher-KPP and Allen-Cahn type.
\end{abstract}


\section{Introduction}
Numerous applications in mathematics, science and engineering require the computation of a matrix function $f(\mathbf{A})$, where $\mathbf{A}\in\mathbb{R}^{n\times n}$ and $f$ is a function mapping $\mathbb{R}^{n\times n}$ to $\mathbb{R}^{n\times n}$ that generalises an underlying scalar function. Common examples include, $f(x) = e^{x}$ (matrix exponential), $f(x) = \sqrt{x}$ (matrix square root), $f(x) = \sin(x)$ and $f(x)=\cos(x)$ (matrix sine and cosine), and $f(x) = \log(x)$ (matrix logarithm) with applications in areas such as control theory, multivariate statistics, fractional calculus, differential equations, complex networks and Markov models \cite{higham_2008,simpson_2010,carr_2025,estrada_2010}. While multiple equivalent definitions of $f(\mathbf{A})$ exist \cite{higham_2008}, if $f(x)$ possesses a convergent power series representation, a simple way to define $f(\mathbf{A})$ is to replace the scalar powers of $x$ with the corresponding matrix powers of $\mathbf{A}$. For example, given $\smash{e^{x} = \sum_{k=0}^{\infty} \frac{1}{k!}x^{k}}$ for all $x\in\mathbb{R}$, the matrix exponential can be defined by the power series $\smash{e^{\mathbf{A}} = \sum_{k=0}^{\infty}\frac{1}{k!}\mathbf{A}^{k}}$ for all $\mathbf{A}\in\mathbb{R}^{n\times n}$ \cite{higham_2008}. Apart from perhaps the matrix inverse (i.e. $f(x) = x^{-1}$), the most studied matrix function is the matrix exponential \cite{wu_2011,ding_2013,ashi_2009,moler_2003} mostly due to the fundamental role it plays in differential equations, e.g. $\mathbf{u}(t) = e^{t\mathbf{A}}\mathbf{u}_{0}$ exactly solves $\mathbf{u}^{\prime}(t) = \mathbf{A}\mathbf{u}$ subject to $\mathbf{u}(0) = \mathbf{u}_{0}$.

Exponential integrators are a class of explicit methods for obtaining approximate solutions to systems of differential equations, distinguished from standard explicit methods due to their use of the matrix exponential and other exponential-like matrix functions \cite{hochbruck_2010,minchev_2005}. These methods have emerged as a strong competitor to traditional implicit methods for stiff problems i.e. problems where the backward Euler method significantly outperforms the forward Euler method \cite{hochbruck_2010}. One of the key reasons for this is that exponential propagators possess favourable properties compared to the linear and shifted inverse propagators of the forward and backward Euler methods \cite{hochbruck_2010}. Overall, exponential integrators enjoy the best of both explicit and implicit time discretisation methods, allowing for large time steps like implicit methods and avoiding the solution of systems of nonlinear equations like explicit methods.

The exponential-like matrix functions appearing in exponential integrators are the so-called matrix phi-functions \cite{skaflestad_2009,higham_2008,hochbruck_2010,aceto_2025,al-mohy_2026}, which emerge in exact solutions of certain linear systems of differential equations, e.g., $\mathbf{u}(t)=~e^{t\mathbf{A}}\mathbf{u}_{0} + t\varphi_{1}(t\mathbf{A})\mathbf{b}$ exactly solves $\mathbf{u}^{\prime}(t) = \mathbf{A}\mathbf{u} + \mathbf{b}$ subject to $\mathbf{u}(0) = \mathbf{u}_{0}$. For scalar argument, $\varphi_{\ell}(x)$ is defined by the power series $\smash{\varphi_{\ell}(x) = \sum_{k=0}^{\infty} \frac{1}{(k+\ell)!}x^{k}}$ for all $x\in\mathbb{R}$ and $\ell = 0,\hdots,p$, where $p$ is a positive integer. With $\varphi_{0}(x) = e^{x}$, the phi-functions satisfy the recurrence relationship $\smash{x\varphi_{\ell}(x) = \varphi_{\ell-1}(x)-\frac{1}{(\ell-1)!}}$ for $\ell = 1,\hdots,p$, which for nonzero $x$ produces the explicit exponential-like functions $\varphi_{1}(x) = (e^{x}-1)/x$, $\varphi_{2}(x) = (e^{x}-1-x)/x^{2}$ and in general $\varphi_{\ell}(x) = (e^{x}-P_{\ell-1}(x))/x^{\ell}$ where $P_{\ell-1}(x)$ is the $(\ell-1)$th degree Maclaurin series of $e^{x}$ \cite{blanes_2026,ashi_2009}. In the matrix setting, $\varphi_{\ell}(\mathbf{A})$ is well defined by the power series 
\begin{gather}
\label{eq:power_series}
\varphi_{\ell}(\mathbf{A}) = \sum_{k=0}^{\infty}\frac{1}{(k+\ell)!}\mathbf{A}^{k},\quad\ell = 0,\hdots,p,
\end{gather}
for all $\mathbf{A}\in\mathbb{R}^{n\times n}$ \cite{schmelzer_2008,wu_2018}, and the recurrence relationship becomes:
\begin{gather}
\label{eq:recurrence_relation}
\mathbf{A}\varphi_{\ell}(\mathbf{A}) = \varphi_{\ell-1}(\mathbf{A})-\frac{1}{(\ell-1)!}\mathbf{I},\quad \ell = 1,\hdots,p,
\end{gather}
where $\mathbf{I}$ is the $n\times n$ identity matrix and $\varphi_{0}(\mathbf{A}) = e^{\mathbf{A}}$. In a similar way to the scalar setting, the recurrence relationship (\ref{eq:recurrence_relation}) yields, this time for nonsingular $\mathbf{A}$, closed-form expressions in terms of the matrix exponential: $\varphi_{1}(\mathbf{A}) = \mathbf{A}^{-1}(e^{\mathbf{A}}-\mathbf{I})$, $\varphi_{2}(\mathbf{A}) = \mathbf{A}^{-2}(e^{\mathbf{A}}-\mathbf{I}-\mathbf{A})$ and in general 
\begin{gather}
\label{eq:closed_form}
\varphi_{\ell}(\mathbf{A}) = \mathbf{A}^{-\ell}(e^{\mathbf{A}}-P_{\ell-1}(\mathbf{A})),\quad \ell = 1,\hdots,p.
\end{gather} 
Due to their prevalence in solutions of ODE systems, a considerable amount of literature has focussed on computing the matrix phi-functions and their action on vectors (e.g. \cite{skaflestad_2009,jawecki_2022,carr_2010,al-mohy_2026,caliari_2023}).  

Systems of differential equations amenable to exact solution or approximate solution using the matrix phi-functions occur in many applications (e.g. \cite{carr_2011,carr_2013,pabisz_2025}). One of the primary reasons for this is that such systems commonly arise from the spatial discretisation of linear and semi-linear partial differential equations (PDEs) where the matrix $\mathbf{A}\in\mathbb{R}^{n\times n}$ represents the discretised linear term and $n$ is the number of mesh/grid nodes. In such problems, $\mathbf{A}$ is often (or can often be transformed into) a transition-rate matrix (or intensity matrix \cite{higham_2008} or infinitesimal generator \cite{sidje_1998}), a matrix whose off-diagonal entries are all nonnegative and whose rows or columns (or both) each sum to zero (Definition \ref{def:1}(v)(vi)). For example, when spatially-discretising PDEs representing conservation laws using the finite volume method, transition-rate matrices arise due to the local conservation satisfied at interfaces between finite volumes with the flux exiting a finite volume through an interface balanced by the flux entering the finite volume sharing that interface. A similar analogy occurs in Markov chains, where conservation requires that the arrival rate into a given state be balanced by the departure rates from all other states. Note that, since a matrix with zero row (and/or column) sums is singular, when $\mathbf{A}$ is a transition-rate matrix $\varphi_{\ell}(\mathbf{A})$ can be defined by the power series (\ref{eq:power_series}) but not the closed-form expression (\ref{eq:closed_form}).

Two fundamental results concerning the matrix exponential include $e^{\mathbf{A}}$ is nonnegative if and only if $\mathbf{A}$ is essentially nonnegative (\cite{higham_2008,bellman_1980}; see Definition \ref{def:1}(i)--(ii)) and $e^{\mathbf{A}}$ is stochastic if $\mathbf{A}$ is a transition-rate matrix (\cite{higham_2008,carr_2025}; see Definition \ref{def:1}(v)--(viii)). In this paper, we study some related properties for the matrix phi-functions including an additional result on the diagonal entries of $\varphi_{\ell}(\mathbf{A})$ that holds for the matrix phi-functions ($\ell = 1,\hdots,p$) but not the matrix exponential itself ($\ell=0$) (Section \ref{sec:properties}). We then use this theory to generate insights into exact and approximate solutions of some linear and semi-linear ODE systems and show how such solutions can be used to generate stochastic simulations of the corresponding deterministic model (Section \ref{sec:application}). Finally, we conclude the paper with some illustrative examples arising from spatial discretisation of some common PDE models of diffusion, advection-diffusion, Fisher-KPP and Allen-Cahn type (Section \ref{sec:examples}).

\section{Properties}
\label{sec:properties}
In this paper, we use $[f(\mathbf{A})]_{ij}$ to denote the entry of $f(\mathbf{A})$ in the $i$th row and $j$th column (or $[f(\mathbf{A})]_{i,j}$ when the indices are more complicated). To be consistent with this notation, we also use $[\mathbf{A}]_{ij}$ and $[\mathbf{A}]_{i,j}$ to denote the entry of $\mathbf{A}$ in the $i$th row and $j$th column and $[\mathbf{b}]_{i}$ to denote the $i$th entry of vector $\mathbf{b}$.

\begin{definition} \label{def:1}\normalfont
Let $\mathbf{A}\in\mathbb{R}^{n\times n}$.
\begin{enumerate}[(i),itemsep=-1pt,topsep=0pt]
\item $\mathbf{A}$ is \textit{nonnegative} if $[\mathbf{A}]_{ij}\geq 0$ for all $i = 1,\hdots,n$ and $j=1,\hdots,n$.
\item $\mathbf{A}$ is \textit{essentially nonnegative} if $[\mathbf{A}]_{ij}\geq 0$ for all $i\neq j$.
\item The \textit{$i$th row sum} of $\mathbf{A}$ is defined by $R_{i}(\mathbf{A}) = \sum_{j=1}^{n} [\mathbf{A}]_{ij}$.
\item The \textit{$j$th column sum} of $\mathbf{A}$ is defined by $C_{j}(\mathbf{A}) = \sum_{i=1}^{n} [\mathbf{A}]_{ij}$.
\item $\mathbf{A}$ is a \textit{row transition-rate matrix} if $\mathbf{A}$ is essentially nonnegative and $R_{i}(\mathbf{A}) = 0$ for all $i = 1,\hdots,n$. Note that this implies that $[\mathbf{A}]_{ii} = -\sum_{j\neq i}[\mathbf{A}]_{ij}\leq 0$.
\item $\mathbf{A}$ is a \textit{column transition-rate matrix} if $\mathbf{A}$ is essentially nonnegative and $C_{j}(\mathbf{A}) = 0$ for all $j = 1,\hdots,n$. Note that this implies that $[\mathbf{A}]_{jj} = -\sum_{i\neq j}[\mathbf{A}]_{ij}\leq 0$.
\item $\mathbf{A}$ is a \textit{row stochastic matrix} if $\mathbf{A}$ is nonnegative and $R_{i}(\mathbf{A}) = 1$ for all $i=1,\hdots,n$.
\item $\mathbf{A}$ is a \textit{column stochastic matrix} if $\mathbf{A}$ is nonnegative and $C_{j}(\mathbf{A}) = 1$ for all $j=1,\hdots,n$.
\end{enumerate}
\end{definition}
\smallskip

\begin{proposition}\label{prop:1}\normalfont
Let $\mathbf{A}\in\mathbb{R}^{n\times n}$ and $p$ be a positive integer.
\begin{enumerate}[(i),itemsep=-1pt,topsep=0pt]
\item If $\mathbf{A}$ is essentially nonnegative, then $\varphi_{\ell}(\mathbf{A})$ is nonnegative for all $\ell = 0,\hdots,p$.
\item If $\mathbf{A}$ is essentially nonnegative, then $\varphi_{\ell}(\mathbf{A}) - \ell\varphi_{\ell+1}(\mathbf{A})$ is nonnegative for all $\ell = 0,\hdots,p$. 
\end{enumerate}
\end{proposition}

\noindent\textbf{Proof.}\\[-3ex]
\begin{enumerate}[(i),itemsep=0pt,topsep=0pt]
\item Consider the identity
\begin{gather}
\label{eq:expB}
\mathbf{B} = \left[\begin{matrix} \mathbf{A} & \mathbf{I} & \mathbf{0} & \hdots & \mathbf{0}\\
 & \mathbf{0} & \mathbf{I} & \ddots & \vdots\\
 &  & \mathbf{0} & \ddots & \mathbf{0}\\
 &  &  & \ddots & \mathbf{I}\\
  &  &  &  & \mathbf{0}
\end{matrix}\right]\quad
\text{then}\quad 
e^{\mathbf{B}} = \left[\begin{matrix} e^{\mathbf{A}} & \varphi_{1}(\mathbf{A}) & \varphi_{2}(\mathbf{A}) & \hdots & \varphi_{p}(\mathbf{A})\\
 & \mathbf{I} & \frac{1}{1!}\mathbf{I} & \hdots & \frac{1}{(p-1)!}\mathbf{I}\\
 &  & \mathbf{I} & \ddots & \vdots\\
 &  &  & \ddots & \frac{1}{1!}\mathbf{I}\\
  &  &  &  & \mathbf{I}
\end{matrix}\right],
\end{gather}
which is stated in \cite{jimenez_2020} and can be shown using a fairly straightforward generalisation of earlier results \cite{sidje_1998,saad_1992,minchev_2004}. A generalised result was also very recently provided in \cite{al-mohy_2026} for an arbitrary analytic function in place of the exponential. Now, if $\mathbf{A}$ is essentially nonnegative then clearly $\mathbf{B}$ is too. Since $\mathbf{B}$ is essentially nonnegative, $e^{\mathbf{B}}$ is nonnegative \cite[Thm 10.29]{higham_2008} and hence from (\ref{eq:expB}) it is clear that $\varphi_{\ell}(\mathbf{A})$ is nonnegative for all $\ell = 0,\hdots,p$.
\item The result is trivially true for $\ell = 0$ since $\varphi_{0}(\mathbf{A}) = e^{\mathbf{A}}$ is nonnegative if $\mathbf{A}$ is essentially nonnegative \cite[Thm 10.29]{higham_2008}. For $\ell = 1,\hdots,p$, the integral representation of the matrix phi-functions \cite{al-mohy_2026,skaflestad_2009,hochbruck_2010} 
\begin{gather}
\label{eq:phi_integral}
\varphi_{\ell}(\mathbf{A}) = \frac{1}{(\ell-1)!}\int_{0}^{1}\theta^{\ell-1}e^{(1-\theta)\mathbf{A}}\,\text{d}\theta,
\end{gather}
gives the integral expression
\begin{align}
\label{eq:varphi_property}
\varphi_{\ell}(\mathbf{A}) - \ell\varphi_{\ell+1}(\mathbf{A}) &= \frac{1}{(\ell-1)!}\int_{0}^{1} \bigl[\theta^{\ell-1} - \theta^{\ell}\bigr]e^{(1-\theta)\mathbf{A}}\,\text{d}\theta.
\end{align}
Now, if $\mathbf{A}$ is essentially nonnegative, then $e^{(1-\theta)\mathbf{A}}$ is nonnegative for all $\theta\in[0,1]$. Since $\theta^{\ell-1}-\theta^{\ell}$ and $1-\theta$ are also both nonnegative for all $\theta\in[0,1]$, it is clear that (\ref{eq:varphi_property}) is nonnegative for all $\ell = 1,\hdots.p$. Note that equation (\ref{eq:phi_integral}) provides an alternative way to prove Proposition \ref{prop:1}(i) since (\ref{eq:phi_integral}) is clearly nonnegative when $\mathbf{A}$ is essentially nonnegative. \qed
\end{enumerate}

\medskip

\begin{proposition}\label{prop:2}\normalfont
Let $\mathbf{A}\in\mathbb{R}^{n\times n}$ be a row transition-rate matrix and $p$ be a positive integer.
\begin{enumerate}[(i),itemsep=-1pt,topsep=0pt]
\item $\widetilde{\varphi}_{\ell}(\mathbf{A}) := \ell!\varphi_{\ell}(\mathbf{A})$ is a row stochastic matrix for all $\ell = 0,\hdots,p$.
\item $[\varphi_{\ell}(\mathbf{A})]_{ij}\in [0,\frac{1}{\ell!}]$ for all $i = 1,\hdots,n$ and $j = 1,\hdots,n$ and for all $\ell = 0,\hdots,p$.
\item $[\varphi_{\ell}(\mathbf{A})]_{jj}\geq [\varphi_{\ell}(\mathbf{A})]_{ij}$ for all $i = 1,\hdots,n$ and $j = 1,\hdots,n$ and for all $\ell = 1,\hdots,p$.
\end{enumerate}
\end{proposition}

\noindent\textbf{Proof.}\\[-3ex]
\begin{enumerate}[(i),itemsep=0pt,topsep=0pt,leftmargin=0.8cm]
\item If $\mathbf{A}$ is a row transition-rate matrix, then $\mathbf{A}$ is essentially nonnegative (Definition \ref{def:1}(v)) and therefore $\varphi_{\ell}(\mathbf{A})$ is nonnegative for all $\ell = 0,\hdots,p$ (Proposition \ref{prop:1}(i)). Since $\ell!$ is positive,  clearly $\widetilde{\varphi}_{\ell}(\mathbf{A})$ is also nonnegative for all $\ell = 0,\hdots,p$. Now, as $\mathbf{A}$ is a row transition-rate matrix, $R_{i}(\mathbf{A})=0$ for all $i=1,\hdots,n$, and the $i$th row sum of $\widetilde{\varphi}_{\ell}(\mathbf{A})$ is given by
\begin{gather*}
R_{i}(\widetilde{\varphi}_{\ell}(\mathbf{A})) = \ell! R_{i}(\varphi_{\ell}(\mathbf{A})) = \ell!\left[\frac{1}{\ell!}R_{i}(\mathbf{I}) + \sum_{k=1}^{\infty} \frac{1}{(k+\ell)!}R_{i}(\mathbf{A}^{k})\right] = 1,
\end{gather*}
for all $i = 1,\hdots,n$. Here, we have used the power series expansion (\ref{eq:power_series}) and the fact that $R_{i}(\mathbf{A}^{k}) = 0$  for all $k\in\{1,2,\hdots\}$ when $R_{i}(\mathbf{A}) = 0$ for all $i=1,\hdots,n$ (i.e., $\mathbf{A}\mathbf{e}=\mathbf{0}$ implies $\mathbf{A}^{k}\mathbf{e} = \mathbf{A}^{k-1}(\mathbf{A}\mathbf{e}) = \mathbf{A}^{k-1}\mathbf{0} = \mathbf{0}$, where $\mathbf{e} = (1,\hdots,1)^{T}$). Since $\widetilde{\varphi}_{\ell}(\mathbf{A})$ is nonnegative and the sum of each row of $\widetilde{\varphi}_{\ell}(\mathbf{A})$ is equal to one, $\widetilde{\varphi}_{\ell}(\mathbf{A})$ is a row stochastic matrix.
\item Follows directly from the fact that $\varphi_{\ell}(\mathbf{A})$ is nonnegative and $R_{i}(\varphi_{\ell}(\mathbf{A})) = \frac{1}{\ell!}$ for all $i=1,\hdots,n$.

\item Recall the recurrence relation (\ref{eq:recurrence_relation}), restated here for convenience: $\smash{\mathbf{A}\varphi_{\ell}(\mathbf{A}) = \varphi_{\ell-1}(\mathbf{A}) - \frac{1}{(\ell-1) !}\mathbf{I}}$ for $\ell = 1,\hdots,p$. From (ii), we have $\smash{[\varphi_{\ell-1}(\mathbf{A})]_{ij}\in[0,\frac{1}{(\ell-1)!}]}$ and therefore the diagonal entries of $\smash{\varphi_{\ell-1}(\mathbf{A}) - \frac{1}{(\ell-1) !}\mathbf{I}}$ are nonpositive and the off-diagonal entries of $\smash{\varphi_{\ell-1}(\mathbf{A}) - \frac{1}{(\ell-1) !}\mathbf{I}}$ are nonnegative. Next, consider the off-diagonal entries of $\mathbf{A}\varphi_{\ell}(\mathbf{A})$
\begin{align*}
[\mathbf{A}\varphi_{\ell}(\mathbf{A})]_{ij} = \sum_{k=1}^{n}[\mathbf{A}]_{ik}[\varphi_{\ell}(\mathbf{A})]_{kj} &= [\mathbf{A}]_{ii}[\varphi_{\ell}(\mathbf{A})]_{ij} + \sum_{k\neq i} [\mathbf{A}]_{ik}[\varphi_{\ell}(\mathbf{A})]_{kj}\\
&= -\sum_{k\neq i} [\mathbf{A}]_{ik}[\varphi_{\ell}(\mathbf{A})]_{ij} + \sum_{k\neq i} [\mathbf{A}]_{ik}[\varphi_{\ell}(\mathbf{A})]_{kj}\\
&= \sum_{k\neq i} [\mathbf{A}]_{ik}\left([\varphi_{\ell}(\mathbf{A})]_{kj} - [\varphi_{\ell}(\mathbf{A})]_{ij}\right),
\end{align*}
where we have used the fact that $\mathbf{A}$ is row transition-rate matrix with $[\mathbf{A}]_{ii} = -\sum_{k\neq i} [\mathbf{A}]_{ik}$ (Definition \ref{def:1}(v)). Since $[\mathbf{A}]_{ik}$ is nonnegative for all $k\neq i$, a necessary condition for $[\mathbf{A}\varphi_{\ell}(\mathbf{A})]_{ij}$ to be nonnegative is that for all $i\neq j$ there exists at least one $k\neq i$ such that $[\varphi_{\ell}(\mathbf{A})]_{ij}\leq [\varphi_{\ell}(\mathbf{A})]_{kj}$. Similarly, the diagonal entries of $\mathbf{A}\varphi_{\ell}(\mathbf{A})$ can be expressed as
\begin{align*}
[\mathbf{A}\varphi_{\ell}(\mathbf{A})]_{jj} = \sum_{k\neq j} [\mathbf{A}]_{jk}\left([\varphi_{\ell}(\mathbf{A})]_{kj} - [\varphi_{\ell}(\mathbf{A})]_{jj}\right).
\end{align*}
Since $[\mathbf{A}]_{jk}$ is nonnegative for all $k\neq j$, a necessary condition for $[\mathbf{A}\varphi_{\ell}(\mathbf{A})]_{jj}$ to be nonpositive is that for all $j$ there exists at least one $k\neq j$ such that $[\varphi_{\ell}(\mathbf{A})]_{kj}\leq  [\varphi_{\ell}(\mathbf{A})]_{jj}$. In summary, in the $j$th column of $\varphi_{\ell}(\mathbf{A})$, each off-diagonal entry must be less than or equal to at least one other entry in the $j$th column and each diagonal entry must be greater than at least one other entry in the $j$th column.  Clearly, the only way this can be true is if every off-diagonal entry in a given column is less than or equal to the diagonal entry i.e. $ [\varphi_{\ell}(\mathbf{A})]_{ij}\leq[\varphi_{\ell}(\mathbf{A})]_{jj}$. \qed
\end{enumerate}

\bigskip
\begin{proposition}\label{prop:3}\normalfont
Let $\mathbf{A}\in\mathbb{R}^{n\times n}$ be a column transition-rate matrix and $p$ be a positive integer.
\begin{enumerate}[(i),itemsep=0pt,topsep=0pt,leftmargin=0.8cm]
\item $\widetilde{\varphi}_{\ell}(\mathbf{A}) := \ell!\varphi_{\ell}(\mathbf{A})$ is a column stochastic matrix for all $\ell = 0,\hdots,p$.
\item $[\varphi_{\ell}(\mathbf{A})]_{ij}\in [0,\frac{1}{\ell!}]$ for all $i = 1,\hdots,n$ and $j = 1,\hdots,n$ and for all $\ell = 0,\hdots,p$.
\item $[\varphi_{\ell}(\mathbf{A})]_{ii}\geq [\varphi_{\ell}(\mathbf{A})]_{ij}$ for all $i = 1,\hdots,n$ and $j = 1,\hdots,n$ and for all $\ell = 1,\hdots,p$. 
\end{enumerate}
\end{proposition}

\noindent\textbf{Proof.} All three results follow from Proposition \ref{prop:2}, the identities $f(\mathbf{A}^{T}) = f(\mathbf{A})^{T}$ and $f(\mathbf{A}) = f(\mathbf{A}^{T})^{T}$ \cite[Thm 1.13(b)]{higham_2008}, and the observations that $\mathbf{A}^{T}$ is a row transition-rate matrix if $\mathbf{A}$ is a column transition-rate matrix and $f(\mathbf{A})^{T}$ is a column stochastic matrix if $f(\mathbf{A})$ is a row stochastic matrix.
\begin{enumerate}[(i),itemsep=0pt,topsep=0pt,leftmargin=0.8cm]
\item If $\mathbf{A}$ is a column transition-rate matrix then $\widetilde{\varphi}_{\ell}(\mathbf{A}^{T})$ is a row stochastic matrix (Proposition \ref{prop:2}(i)) and hence $\widetilde{\varphi}_{\ell}(\mathbf{A}) = \widetilde{\varphi}_{\ell}(\mathbf{A}^{T})^{T}$ is a column stochastic matrix.
\item $\varphi_{\ell}(\mathbf{A})$ is nonnegative from Proposition \ref{prop:1}(i) and $[\varphi_{\ell}(\mathbf{A})]_{ij} = [\varphi_{\ell}(\mathbf{A})^{T}]_{ji} = [\varphi_{\ell}(\mathbf{A}^{T})]_{ji}\leq\frac{1}{\ell!}$.
\item $[\varphi_{\ell}(\mathbf{A})]_{ii} = [\varphi_{\ell}(\mathbf{A})^{T}]_{ii} = [\varphi_{\ell}(\mathbf{A}^{T})]_{ii}\geq [\varphi_{\ell}(\mathbf{A}^{T})]_{ji} = [\varphi_{\ell}(\mathbf{A})]_{ij}$. \qed
\end{enumerate}

\bigskip
\begin{remark}\normalfont\label{rem:1}
Let $\gamma$ be a nonnegative scalar. If $\mathbf{A}$ is essentially nonnegative then $\gamma\mathbf{A}$ is also essentially nonnegative. Similarly, if $\mathbf{A}$ is a row (column) transition-rate matrix then $\gamma\mathbf{A}$ is also a row (column) transition-rate matrix. Therefore, Propositions \ref{prop:1}--\ref{prop:3} remain true if $\mathbf{A}$ is replaced by $\gamma\mathbf{A}$. This is important in Section \ref{sec:application}, where $\gamma$ represents positive time or a positive time step.
\end{remark}

\begin{remark}\normalfont\label{rem:2}
Note that Propositions \ref{prop:2}(iii) and \ref{prop:3}(iii) do not include $\ell=0$ as both results are not true in general for the matrix exponential i.e. $\varphi_{0}(\mathbf{A}) = e^{\mathbf{A}}$. For example, the exponential of the row transition-rate matrix
\begin{gather*}
\mathbf{A} = \begin{bmatrix*}[r] -1 & 1 & 0\\ 1 & -2 & 1\\ 0 & 2 & -2\end{bmatrix*},
\end{gather*}
is a matrix whose largest entry in the second column is not the diagonal entry:
\begin{gather*}
e^{\mathbf{A}} = \left[\begin{matrix} \eh{0.5335} & 0.3444 & 0.1221\\ 0.3444 & 0.4332 & 0.2224\\ 0.2441 & \eh{0.4447} & \eh{0.3112}\end{matrix}\right] = \varphi_{0}(\mathbf{A}).
\end{gather*}
In contrast, the phi-functions, of which the first six are given below,
\begin{gather*}
\varphi_{1}(\mathbf{A}) = \left[\begin{matrix} \eh{0.7043} & 0.2378 & 0.0579\\ 0.2378 & \eh{0.5822} & 0.1799\\ 0.1158 & 0.3599 & \eh{0.5244}\end{matrix}\right],\quad 
\varphi_{2}(\mathbf{A}) = \left[\begin{matrix} \eh{0.3890} & 0.0933 & 0.0177\\ 0.0933 & \eh{0.3311} & 0.0756\\ 0.0354 & 0.1512 & \eh{0.3134}\end{matrix}\right],\\
\varphi_{3}(\mathbf{A}) = \left[\begin{matrix} \eh{0.1368} & 0.0258 & 0.0040\\ 0.0258 & \eh{0.1191} & 0.0218\\ 0.0081 & 0.0435 & \eh{0.1151}\end{matrix}\right],\quad 
\varphi_{4}(\mathbf{A}) = \left[\begin{matrix} \eh{0.0354} & 0.0055 & 0.0007\\ 0.0055 & \eh{0.0313} & 0.0048\\ 0.0015 & 0.0096 & \eh{0.0306}\end{matrix}\right],\\
\varphi_{5}(\mathbf{A}) = \left[\begin{matrix} \eh{0.0072} & 0.0010 & 0.0001\\ 0.0010 & \eh{0.0065} & 0.0009\\ 0.0002 & 0.0017 & \eh{0.0064}\end{matrix}\right],\quad 
\varphi_{6}(\mathbf{A}) = \left[\begin{matrix} \eh{0.0012} & 0.0001 & 0.0000\\ 0.0001 & \eh{0.0011} & 0.0001\\ 0.0000 & 0.0003 & \eh{0.0011}\end{matrix}\right],
\end{gather*}
all satisfy the property in Proposition \ref{prop:2}(iii) as expected. Additionally, the proof of Proposition \ref{prop:2}(iii) breaks down for the matrix exponential as the recurrence relation (\ref{eq:recurrence_relation}) does not hold when $\ell = 0$.
\end{remark}

\begin{remark}\normalfont\label{rem:3}
While not true in general, Propositions \ref{prop:2}(iii) and \ref{prop:3}(iii) hold for the matrix exponential in some special cases. For example, the general $2\times 2$ transition-rate matrix:
\begin{equation*}
\mathbf{A} = \left[\begin{matrix*}[r] -a & a\\ a & -a\end{matrix*}\right],\quad e^{\mathbf{A}} = \frac{1}{2}\left[\begin{matrix*}[r] 1 + e^{-2a} & 1 - e^{-2a}\\ 1-e^{-2a} & 1 + e^{-2a}\end{matrix*}\right],
\end{equation*}
satisfies both $[e^{\mathbf{A}}]_{jj}\geq [e^{\mathbf{A}}]_{ij}$ and $[e^{\mathbf{A}}]_{ii}\geq [e^{\mathbf{A}}]_{ij}$ for all $i=1,2$ and $j=1,2$. 
\end{remark}

\section{Application}
\label{sec:application}
As mentioned earlier, the matrix phi-functions appear in exact solutions of systems of differential equations and in exponential integrators for obtaining approximate solutions to systems of differential equations. Using the properties of the previous section, we now provide some insights into deterministic solutions of certain linear and semi-linear systems of differential equations and then discuss how these can be used to generate stochastic simulations. Throughout this section we employ the following notation: $\mathbf{u}\in\mathbb{R}^{n}$ is the solution/state vector, $t>0$ is time, $\mathbf{u}_{0}$ is the specified initial vector, $\mathbf{A}$ is a constant $n\times n$ matrix and $\mathbf{b}(\cdot)$ is a forcing term, which is either constant or a function of $t$ or $\mathbf{u}$.

\subsection{Deterministic model: constant or time-dependent polynomial forcing term}
\label{sec:ODE_systems1}
The system of linear differential equations with constant forcing term
\begin{gather}
\label{eq:ode_constant}
\frac{\text{d}\mathbf{u}}{\text{d}t} = \mathbf{A}\mathbf{u} + \mathbf{b},\qquad \mathbf{u}(0) = \mathbf{u}_{0},
\end{gather}
where $\mathbf{b}$ is a constant vector, has the well-known exact solution
\begin{gather}
\label{eq:exact_solution_constant}
\mathbf{u}(t) = e^{t\mathbf{A}}\mathbf{u}_{0} + t\varphi_{1}(t\mathbf{A})\mathbf{b}.
\end{gather}
This result can also be generalised to a system of linear differential equations with a time-dependent polynomial forcing term \cite{minchev_2005}:
\begin{gather}
\label{eq:ode_polynomial}
\frac{\text{d}\mathbf{u}}{\text{d}t} = \mathbf{A}\mathbf{u} + \sum_{\ell=1}^{p}t^{\ell-1}\mathbf{b}_{\ell},\qquad \mathbf{u}(0) = \mathbf{u}_{0},
\end{gather}
where each $\mathbf{b}_{\ell}\in\mathbb{R}^{n}$ is a constant vector and $p$ is a nonnegative integer. Here, if $p=0$, we adopt the convention that the summation in (\ref{eq:ode_polynomial}) evaluates to the zero vector while $p=1$ recovers (\ref{eq:ode_constant}). Applying the integrating factor $e^{t\mathbf{A}}$ to (\ref{eq:ode_polynomial}) leads to the exact solution:
\begin{gather}
\label{eq:exact_solution_if}
\mathbf{u}(t) = e^{t\mathbf{A}}\mathbf{u}_{0} + \sum_{\ell=1}^{p}\left[\int_{0}^{t}s^{\ell-1}e^{(t-s)\mathbf{A}}\,\text{d}s\right]\mathbf{b}_{\ell},
\end{gather}
which can be expressed in terms of the matrix phi-functions:
\begin{gather}
\label{eq:exact_solution}
\mathbf{u}(t) = e^{t\mathbf{A}}\mathbf{u}_{0} + \sum_{\ell=1}^{p}(\ell-1)!t^{\ell}\varphi_{\ell}(t\mathbf{A})\mathbf{b}_{\ell},
\end{gather}
using the integral representation \cite{minchev_2005}:
\begin{gather}
\label{eq:phit_integral}
\varphi_{\ell}(t\mathbf{A}) = \frac{1}{t^{\ell}(\ell-1)!}\int_{0}^{t} s^{\ell-1}e^{(t-s)\mathbf{A}}\,\text{d}s,
\end{gather}
which follows from (\ref{eq:phi_integral}) and the change of variable $s = \theta t$.
Note that (\ref{eq:exact_solution}) recovers (\ref{eq:exact_solution_constant}) when $p=1$ and can be expressed equivalently in terms of the scaled matrix $\widetilde{\varphi}$-functions (see Propositions \ref{prop:2} and \ref{prop:3}):
\begin{gather}
\label{eq:exact_solution_scaled}
\mathbf{u}(t) = e^{t\mathbf{A}}\mathbf{u}_{0} + \sum_{\ell=1}^{p}\widetilde{\varphi}_{\ell}(t\mathbf{A})[\tfrac{t^{\ell}}{\ell}\mathbf{b}_{\ell}].
\end{gather}
Based on the above results and Propositions \ref{prop:1}--\ref{prop:3} we can draw the following insights:
\begin{enumerate}[(i),itemsep=0pt,topsep=0pt]
\item If $\mathbf{A}$ is essentially nonnegative, then $\mathbf{u}(t)$ (\ref{eq:exact_solution}) is nonnegative if $\mathbf{u}_{0}$ is nonnegative and $\mathbf{b}_{\ell}$ is nonnegative for all $\ell = 1,\hdots,p$ (\ref{eq:exact_solution}). Note that these results are consistent with the well-known result that the solution of $\mathbf{u}^{\prime}(t) = \mathbf{A}\mathbf{u} + \mathbf{b}(t)$ with $\mathbf{u}(0)=\mathbf{u}_{0}$ is nonnegative if $\mathbf{A}$ is essentially nonnegative and both $\mathbf{u}_{0}$ and $\mathbf{b}(t)$ are nonnegative \cite{bellman_1980}. 
\item Consider the $i$th entry/state of $\mathbf{u}(t)$ (\ref{eq:exact_solution}):
\begin{gather*}
[\mathbf{u}(t) ]_{i} = \sum_{j=1}^{n}[e^{t\mathbf{A}}]_{ij}[\mathbf{u}_{0}]_{j} + \sum_{\ell=1}^{p}(\ell-1)!t^{\ell} \sum_{j=1}^{n}[\varphi_{\ell}(t\mathbf{A})]_{ij}[\mathbf{b}_{\ell}]_{j},
\end{gather*}
from which we see that $[\varphi_{\ell}(t\mathbf{A})]_{ij}$ provides a weighting measuring the impact of $[\mathbf{b}_{\ell}]_{j}$ on $[\mathbf{u}(t)]_{i}$. If $\mathbf{A}$ is a row transition-rate matrix then Proposition \ref{prop:2}(iii) means that out of all the entries of $\mathbf{u}(t)$, $[\mathbf{b}_{\ell}]_{j}$ has the most impact on $[\mathbf{u}(t)]_{j}$. On the other hand, if $\mathbf{A}$ is a column transition-rate matrix then Proposition \ref{prop:3}(iii) means that out of of all the entries of $\mathbf{b}_{\ell}$, $[\mathbf{u}(t)]_{i}$ is most impacted by $[\mathbf{b}_{\ell}]_{i}$.
\item Suppose $[\mathbf{u}(t)]_{i}$ represents the amount of some underlying additive quantity $u$ present in state $i$ at time $t$. If $\mathbf{A}$ is a column transition-rate matrix, then $e^{t\mathbf{A}}$ and $\widetilde{\varphi}_{\ell}(t\mathbf{A})$ are column stochastic matrices (Proposition \ref{prop:3}) and from equation (\ref{eq:exact_solution_scaled}) we have that
\begin{align*}
\sum_{i=1}^{n} [\mathbf{u}(t)]_{i} &= \sum_{i=1}^{n}\sum_{j=1}^{n}  [e^{t\mathbf{A}}]_{ij}[\mathbf{u}_{0}]_{j} +\sum_{\ell=1}^{p} \frac{t^{\ell}}{\ell}\sum_{i=1}^{n}\sum_{j=1}^{n}[\widetilde{\varphi}_{\ell}(t\mathbf{A})]_{ij}[\mathbf{b}_{\ell}]_{j}\\
&= \sum_{j=1}^{n}[\mathbf{u}_{0}]_{j} \sum_{i=1}^{n}  [e^{t\mathbf{A}}]_{ij} + \sum_{\ell=1}^{p}\frac{t^{\ell}}{\ell}\sum_{j=1}^{n}[\mathbf{b}_{\ell}]_{j}\sum_{i=1}^{n}[\widetilde{\varphi}_{\ell}(t\mathbf{A})]_{ij}\\ 
&=\sum_{j=1}^{n}[\mathbf{u}_{0}]_{j} + \sum_{\ell=1}^{p}\frac{t^{\ell}}{\ell}\sum_{j=1}^{n}[\mathbf{b}_{\ell}]_{j},
\end{align*}
where we note that the second term is the integral of the forcing term in equation (\ref{eq:ode_polynomial}) from time $0$ to time $t$, that is, $\smash{\sum_{\ell=1}^{p}\frac{t^{\ell}}{\ell}\sum_{j=1}^{n}[\mathbf{b}_{\ell}]_{j} = \sum_{\ell=1}^{p}\sum_{j=1}^{n}[\int_{0}^{t}s^{\ell-1}\mathbf{b}_{\ell}\,\text{d}s]_{j}}$. In words, the total amount of $u$ in the system at time $t$ is given by the total amount of $u$ in the system at time $0$ plus the net amount of $u$ injected into (or removed from) the system between time $0$ and time $t$. For the special case $p = 0$ or $\mathbf{b}_{1} = \cdots = \mathbf{b}_{p} = \mathbf{0}$ (no source/sink term) the total amount of $u$ remains constant and equal to the initial amount for all time $t$, that is, $\smash{\sum_{i=1}^{n} [\mathbf{u}(t)]_{i} = \sum_{i=1}^{n} [\mathbf{u}_{0}]_{i}}$.
\end{enumerate}

\subsection{Deterministic model: time or state dependent forcing term}
\label{sec:ODE_systems2}
Consider a finite time interval $[0,T]$ discretised into $M$ subintervals and let $\tau = T/M>0$ denote the time step. The system of linear differential equations with time dependent forcing term:
\begin{gather}
\label{eq:ode_time_dependent}
\frac{\text{d}\mathbf{u}}{\text{d}t} = \mathbf{A}\mathbf{u} + \mathbf{b}(t),\qquad \mathbf{u}(0) = \mathbf{u}_{0},
\end{gather} 
can be approximated at times $t_{k} = k\tau$ ($k=1,\hdots,M$) by applying the exponential integrator \cite{hochbruck_2010}:
\begin{gather}
\label{eq:exponential_integrator1}
\mathbf{u}_{k+1} = e^{\tau\mathbf{A}}\mathbf{u}_{k} + \tau\varphi_{1}(\tau\mathbf{A})\mathbf{b}(t_{k}) + \tau\varphi_{2}(\tau\mathbf{A})[\mathbf{b}(t_{k+1})-\mathbf{b}(t_{k})],
\end{gather}
for $k = 0,\hdots,M-1$ with $\mathbf{u}_{k}\approx\mathbf{u}(t_{k})$ ($k=1,\hdots,M$). On the other hand, the system of semi-linear differential equations with state dependent forcing term:
\begin{gather}
\label{eq:ode_state_dependent}
\frac{\text{d}\mathbf{u}}{\text{d}t} = \mathbf{A}\mathbf{u} + \mathbf{b}(\mathbf{u}),\qquad \mathbf{u}(0) = \mathbf{u}_{0},
\end{gather} 
can be approximated at the times $t_{k} = k\tau$ ($k=1,\hdots,M$) by applying the exponential integrator \cite{hochbruck_2010,minchev_2005}:
\begin{gather}
\label{eq:exponential_integrator2}
\mathbf{u}_{k+1} = e^{\tau\mathbf{A}}\mathbf{u}_{k} + \tau\varphi_{1}(\tau\mathbf{A})\mathbf{b}(\mathbf{u}_{k}),
\end{gather}
for $k = 0,\hdots,M-1$ with $\mathbf{u}_{k}\approx\mathbf{u}(t_{k})$ ($k=1,\hdots,M$). Note that (\ref{eq:exponential_integrator1}) and (\ref{eq:exponential_integrator2}) can be expressed equivalently in terms of the scaled matrix $\widetilde{\varphi}$-functions (defined in Propositions \ref{prop:2} and \ref{prop:3}), respectively
\begin{align}
\label{eq:exponential_integrator1_scaled}
\mathbf{u}_{k+1} &= e^{\tau\mathbf{A}}\mathbf{u}_{k} + \widetilde{\varphi}_{1}(\tau\mathbf{A})[\tau\mathbf{b}(t_{k})] + \widetilde{\varphi}_{2}(\tau\mathbf{A})[\tfrac{\tau}{2}[\mathbf{b}(t_{k+1})-\mathbf{b}(t_{k})]],\\
\label{eq:exponential_integrator2_scaled}
\mathbf{u}_{k+1} &= e^{\tau\mathbf{A}}\mathbf{u}_{k} + \widetilde{\varphi}_{1}(\tau\mathbf{A})[\tau\mathbf{b}(\mathbf{u}_{k})].
\end{align}
Based on the above results and Propositions \ref{prop:1}--\ref{prop:3} we can draw the following insights:
\begin{enumerate}[(i),itemsep=0pt,topsep=0pt]
\item If $\mathbf{A}$ is essentially nonnegative, $\mathbf{u}_{0}$ is nonnegative, and either $\mathbf{b}(t)$ is nonnegative for all $t\in[0,T]$ (\ref{eq:exponential_integrator1}) or $\mathbf{b}(\mathbf{u})$ is nonnegative for all nonnegative $\mathbf{u}$ (\ref{eq:exponential_integrator2}), then $\mathbf{u}_{k}$ is nonnegative for all $k = 0,1,\hdots,M$. This follows from Proposition \ref{prop:1}(i) and for (\ref{eq:exponential_integrator1}) by rewriting as $\mathbf{u}_{k+1} = e^{\tau\mathbf{A}}\mathbf{u}_{k} + \tau[\varphi_{1}(\tau\mathbf{A})-\varphi_{2}(\tau\mathbf{A})]\mathbf{b}(t_{k}) + \tau\varphi_{2}(\tau\mathbf{A})\mathbf{b}(t_{k+1})$ and using Proposition \ref{prop:1}(ii).
\item Consider the $i$th entry/state of $\mathbf{u}(t)$ from equations (\ref{eq:exponential_integrator1}) and (\ref{eq:exponential_integrator2}), respectively:
\begin{align}
\label{eq:exponential_integrator1_state}
\hspace*{-0.5em}[\mathbf{u}_{k+1}]_{i} &= \sum_{j=1}^{n}[e^{\tau\mathbf{A}}]_{ij}[\mathbf{u}_{k}]_{j} + \tau\sum_{j=1}^{n}[\varphi_{1}(\tau\mathbf{A})]_{ij}[\mathbf{b}(t_{k})]_{j} + \tau\sum_{j=1}^{n}[\varphi_{2}(\tau\mathbf{A})]_{ij}[\mathbf{b}(t_{k+1})-\mathbf{b}(t_{k})]_{j},\\
\hspace*{-0.5em}\label{eq:exponential_integrator2_state}
[\mathbf{u}_{k+1}]_{i} &= \sum_{j=1}^{n}[e^{\tau\mathbf{A}}]_{ij}[\mathbf{u}_{k}]_{j} + \tau\sum_{j=1}^{n}[\varphi_{1}(\tau\mathbf{A})]_{ij}[\mathbf{b}(\mathbf{u}_{k})]_{j}.
\end{align}
Here we see that $[\varphi_{1}(\tau\mathbf{A})]_{ij}$ and $[\varphi_{2}(\tau\mathbf{A})]_{ij}$ provide weightings measuring the impact of individual entries of the forcing term on $[\mathbf{u}_{k+1}]_{i}$. For example, in equation (\ref{eq:exponential_integrator2_state}), $[\varphi_{1}(\tau\mathbf{A})]_{ij}$ provides a weighting measuring the impact of the value of $[\mathbf{b}(\mathbf{u}_{k})]_{j}$ on $[\mathbf{u}_{k+1}]_{i}$. If $\mathbf{A}$ is a row transition-rate matrix then Proposition \ref{prop:2}(iii) means that out of all the entries of $\mathbf{u}_{k+1}$, $[\mathbf{b}(\mathbf{u}_{k})]_{j}$ has the most impact on $[\mathbf{u}_{k+1}]_{j}$. On the other hand, if $\mathbf{A}$ is a column transition-rate matrix then Proposition \ref{prop:3}(iii) means that out of of all the entries of $\mathbf{b}(\mathbf{u}_{k})$, $[\mathbf{u}_{k+1}]_{i}$ is most impacted by $[\mathbf{b}(\mathbf{u}_{k})]_{i}$. Similar insights can also be provided for equation (\ref{eq:exponential_integrator1_state}).
\item Suppose $[\mathbf{u}_{k+1}]_{i}$ represents the amount of some underlying additive quantity $u$ present in state $i$ at time $t$. If $\mathbf{A}$ is a column transition-rate matrix, then $e^{\tau\mathbf{A}}$ and $\widetilde{\varphi}_{\ell}(\tau\mathbf{A})$ are column stochastic matrices (Proposition \ref{prop:3}) and from equations (\ref{eq:exponential_integrator1_scaled}) and (\ref{eq:exponential_integrator2_scaled}) we have, respectively
\begin{align}
\label{eq:exponential_integrator1_conservation}
\sum_{i=1}^{n} [\mathbf{u}_{k+1}]_{i} &=\sum_{j=1}^{n}[\mathbf{u}_{k}]_{j} + \frac{\tau}{2}\sum_{j=1}^{n}[\mathbf{b}(t_{k})+\mathbf{b}(t_{k+1})]_{j},\\
\label{eq:exponential_integrator2_conservation}
\sum_{i=1}^{n}[\mathbf{u}_{k+1} ]_{i} &= \sum_{j=1}^{n}[\mathbf{u}_{k}]_{j} + \tau\sum_{j=1}^{n}[\mathbf{b}(\mathbf{u}_{k})]_{j}.
\end{align}
Note that the second term on the right-hand side of equations (\ref{eq:exponential_integrator1_conservation}) and (\ref{eq:exponential_integrator2_conservation}) approximate the integral from time $t_{k}$ to time $t_{k+1}$ of the forcing term in equations (\ref{eq:ode_time_dependent}) and (\ref{eq:ode_state_dependent}), respectively, that is, $\smash{\frac{\tau}{2}\sum_{j=1}^{n}[\mathbf{b}(t_{k})+\mathbf{b}(t_{k+1})]_{j} \approx \sum_{j=1}^{n}[\int_{t_{k}}^{t_{k+1}}\mathbf{b}(s)\,\text{d}s]_{j}}$ (trapezoidal rule with one sub-interval) and ${\tau\sum_{j=1}^{n}[\mathbf{b}(\mathbf{u}_{k})]_{i} \approx \sum_{j=1}^{n}[\int_{t_{k}}^{t_{k+1}}\mathbf{b}(\mathbf{u}(s))\,\text{d}s]_{j}}$ (left-hand rule with one sub-interval). In words, the total amount of $u$ in the system at time $t_{k+1}$ is given by the total amount of $u$ in the system at time $t_{k}$ plus an approximation to the net amount of $u$ injected into (or removed from) the system between time $t_{k}$ and time $t_{k+1}$. Note that (\ref{eq:exponential_integrator1_scaled}) utilises the decomposition $\frac{\tau}{2}[\mathbf{b}(t_{k})+\mathbf{b}(t_{k+1})]  = \tau\mathbf{b}(t_{k}) + \tfrac{\tau}{2}[\mathbf{b}(t_{k+1})-\mathbf{b}(t_{k})]$.
\end{enumerate}

\begin{figure}[p]
\algbox{
\begin{algorithm}[Stochastic model -- constant or polynomial forcing term]\mbox{}\\
\label{alg:stochastic_model1}
\emph{
\begin{tabular}{@{}l}
Input matrix $\mathbf{A}$ and vectors $\mathbf{u}_{0}$ and $\mathbf{b}_{\ell}$ for $\ell = 1,\hdots,p$, where $p$ is a positive integer\\
Choose stochastic model discretisation parameter $N_{s}$\\
Choose times of interest $t_{1},\hdots,t_{K}$, where $K$ is a positive integer\\
\textbf{for} $k = 1,\hdots,K$\\
\qquad Set $t =t_{k}$ and initialise $\widehat{\mathbf{u}}(t)=\mathbf{0}$\\ 
\qquad Compute $\varphi_{\ell}(t\mathbf{A})$ and set $\widetilde{\varphi}_{\ell}(t\mathbf{A}) = \ell!\varphi_{\ell}(t\mathbf{A})$ for all $\ell = 0,1,\hdots,p$\\
\qquad Cumulative probabilities $[\mathbf{P}_{\ell}]_{0,j} = 0$ and $[\mathbf{P}_{\ell}]_{ij} = \sum_{m=1}^{i}[\widetilde{\varphi}_{\ell}(t\mathbf{A})]_{mj}$ for all $j=1,\hdots,n$ and $\ell = 0,1,\hdots,p$\\
\qquad\textbf{for} $\ell = 0,1,\hdots,p$\\
\qquad\qquad Set $\mathbf{v}=\mathbf{u}_{0}$ if $\ell = 0$ or $\mathbf{v} = t^{\ell}\mathbf{b}_{\ell}/\ell$ otherwise\\
\qquad\qquad\textbf{for} $j = 1,\hdots,n$\\
\qquad\qquad\qquad Set $N([\mathbf{v}]_{j}) = \sgn([\mathbf{v}]_{j})\lceil|N_{s}[\mathbf{v}]_{j}|\rceil$\\ 
\qquad\qquad\qquad Generate random number $r\sim U(0,1)$\\
\qquad\qquad\qquad Find $i$ such that $[\mathbf{P}_{\ell}]_{i-1,j} < r < [\mathbf{P}_{\ell}]_{i,j}$\\
\qquad\qquad\qquad $[\widehat{\mathbf{u}}(t)]_{i} = [\widehat{\mathbf{u}}(t)]_{i} + [\mathbf{v}]_{j}/N([\mathbf{v}]_{j})$\\
\qquad\qquad\textbf{end}\\
\qquad\textbf{end}\\
\textbf{end}\\
Return $\widehat{\mathbf{u}}(t_{k})$ for $k=1,\hdots,K$
\end{tabular}
}
\end{algorithm}}

\noindent
\algbox{
\begin{algorithm}[Stochastic model -- general time varying forcing term]\mbox{}\\
\label{alg:stochastic_model2}
\emph{
\begin{tabular}{@{}l}
Input matrix $\mathbf{A}$, vector $\mathbf{u}_{0}$, vector function $\mathbf{b}(t)$\\
Choose stochastic model discretisation parameter $N_{s}$\\
Choose time duration $T$ and number of time steps $M$\\
Compute time step $\tau = T/M$ and discrete times $t_{k}=k\tau$ for all $k=0,\hdots,M$\\
Compute $\varphi_{\ell}(\tau\mathbf{A})$ and set $\widetilde{\varphi}_{\ell}(\tau\mathbf{A}) = \ell!\varphi_{\ell}(\tau\mathbf{A})$ for all $\ell = 0,1,2$ and initialise $\widehat{\mathbf{u}}_{0}=\mathbf{u}_{0}$\\
Cumulative probabilities $[\mathbf{P}_{\ell}]_{0,j} = 0$ and $[\mathbf{P}_{\ell}]_{ij} = \sum_{m=1}^{i}[\widetilde{\varphi}_{\ell}(\tau\mathbf{A})]_{mj}$ for all $j=1,\hdots,n$ and $\ell = 0,1,2$\\
\textbf{for} $k = 0,\hdots,M-1$\\
\qquad Initialise $\widehat{\mathbf{u}}_{k+1}=\mathbf{0}$\\ 
\qquad\textbf{for} $\ell = 0,1,2$\\
\qquad\qquad Set $\mathbf{v} = \widehat{\mathbf{u}}_{k}$ if $\ell = 0$, $\mathbf{v} = \tau\mathbf{b}(t_{k})$ if $\ell=1$ or $\mathbf{v} = \tau(\mathbf{b}(t_{k+1})-\mathbf{b}(t_{k}))/2$ if $\ell = 2$\\
\qquad\qquad\textbf{for} $j = 1,\hdots,n$\\
\qquad\qquad\qquad Set $N([\mathbf{v}]_{j}) = \sgn([\mathbf{v}]_{j})\lceil|N_{s}[\mathbf{v}]_{j}|\rceil$\\ 
\qquad\qquad\qquad Generate random number $r\sim U(0,1)$\\
\qquad\qquad\qquad Find $i$ such that $[\mathbf{P}_{\ell}]_{i-1,j} < r < [\mathbf{P}_{\ell}]_{i,j}$\\
\qquad\qquad\qquad $[\widehat{\mathbf{u}}_{k+1}]_{i} = [\widehat{\mathbf{u}}_{k+1}]_{i} + [\mathbf{v}]_{j}/N([\mathbf{v}]_{j})$\\
\qquad\qquad\textbf{end}\\
\qquad\textbf{end}\\
\textbf{end}\\
Return $\widehat{\mathbf{u}}_{k}$ for $k=1,\hdots,M$
\end{tabular}
}
\end{algorithm}}

\noindent
\algbox{
\begin{algorithm}[Stochastic model -- general state varying forcing term]\mbox{}\\
\label{alg:stochastic_model3}
\emph{
\begin{tabular}{@{}l}
Input matrix $\mathbf{A}$, vector $\mathbf{u}_{0}$, vector function $\mathbf{b}(\mathbf{u})$\\
Choose stochastic model discretisation parameter $N_{s}$\\
Choose time duration $T$ and number of time steps $M$, and compute time step $\tau = T/M$\\
Compute $\varphi_{\ell}(\tau\mathbf{A})$ and set $\widetilde{\varphi}_{\ell}(\tau\mathbf{A}) = \ell!\varphi_{\ell}(\tau\mathbf{A})$ for all $\ell = 0,1$ and initialise $\widehat{\mathbf{u}}_{0}=\mathbf{u}_{0}$\\
Cumulative probabilities $[\mathbf{P}_{\ell}]_{0,j} = 0$ and $[\mathbf{P}_{\ell}]_{ij} = \sum_{m=1}^{i}[\widetilde{\varphi}_{\ell}(\tau\mathbf{A})]_{mj}$ for all $j=1,\hdots,n$ and $\ell=0,1$\\
\textbf{for} $k = 0,\hdots,M-1$\\
\qquad Initialise $\widehat{\mathbf{u}}_{k+1}=\mathbf{0}$\\ 
\qquad\textbf{for} $\ell = 0,1$\\
\qquad\qquad Set $\mathbf{v} = \widehat{\mathbf{u}}_{k}$ if $\ell=0$ or $\mathbf{v} = \tau\mathbf{b}(\widehat{\mathbf{u}}_{k})$ if $\ell =1$\\
\qquad\qquad\textbf{for} $j = 1,\hdots,n$\\
\qquad\qquad\qquad Set $N([\mathbf{v}]_{j}) = \sgn([\mathbf{v}]_{j})\lceil|N_{s}[\mathbf{v}]_{j}|\rceil$\\ 
\qquad\qquad\qquad Generate random number $r\sim U(0,1)$\\
\qquad\qquad\qquad Find $i$ such that $[\mathbf{P}_{\ell}]_{i-1,j} < r < [\mathbf{P}_{\ell}]_{i,j}$\\
\qquad\qquad\qquad $[\widehat{\mathbf{u}}_{k+1}]_{i} = [\widehat{\mathbf{u}}_{k+1}]_{i} + [\mathbf{v}]_{j}/N([\mathbf{v}]_{j})$\\
\qquad\qquad\textbf{end}\\
\qquad\textbf{end}\\
\textbf{end}\\
Return $\widehat{\mathbf{u}}_{k}$ for $k=1,\hdots,M$
\end{tabular}}
\end{algorithm}}
\end{figure}

\subsection{Stochastic models}
\label{sec:stochastic_models}
Consider the exact solution (\ref{eq:exact_solution_scaled}) and suppose $[\mathbf{u}(t)]_{i}$ can be partitioned into $N([\mathbf{u}(t)]_{i})$ individual discrete units (e.g. cells, particles) that act independently. Similarly, consider the exponential integrators (\ref{eq:exponential_integrator1_scaled}) and (\ref{eq:exponential_integrator2_scaled}) and suppose $[\mathbf{u}_{k}]_{i}$ can be partitioned into $N([\mathbf{u}_{k}]_{i})$ individual discrete units that act independently. In all cases, $N(\mu) = \lceil|N_{s}\mu|\rceil$ is the smallest nonnegative integer greater than $|N_{s}\mu|$, and $N_{s}$ controls the level of discretisation, denoting the number of individual discrete units when $\mu = 1$. If $\mathbf{A}$ is a column transition-rate matrix, then $e^{t\mathbf{A}}$ and $\widetilde{\varphi}_{\ell}(t\mathbf{A})$ are column stochastic matrices for all $t>0$ and the exact solution (\ref{eq:exact_solution_scaled}) can be used to generate continuous-time stochastic simulations of the deterministic model (\ref{eq:ode_polynomial}). Similarly, if $\mathbf{A}$ is a column transition-rate matrix, then $e^{\tau\mathbf{A}}$ and $\widetilde{\varphi}_{\ell}(\tau\mathbf{A})$ are column stochastic matrices for all $\tau>0$ and the exponential integrators (\ref{eq:exponential_integrator1_scaled}) and (\ref{eq:exponential_integrator2_scaled}) can be used to generate discrete-time stochastic simulations of the deterministic models (\ref{eq:ode_time_dependent}) and (\ref{eq:ode_state_dependent}), respectively. These three processes are described in Algorithms \ref{alg:stochastic_model1}, \ref{alg:stochastic_model2} and \ref{alg:stochastic_model3}, respectively, where $\widehat{\mathbf{u}}(t)$ represents a stochastic realisation of $\mathbf{u}(t)$ (Algorithm 1), $\widehat{\mathbf{u}}_{k}$ represents a stochastic realisation of $\mathbf{u}_{k}$ (Algorithms 2 and 3), and the entries of $[e^{t\mathbf{A}}]_{ij}$ and $[\widetilde{\varphi}_{\ell}(t\mathbf{A})]_{ij}$ (Algorithm 1) and $[e^{\tau\mathbf{A}}]_{ij}$ and $[\widetilde{\varphi}_{\ell}(\tau\mathbf{A})]_{ij}$ (Algorithms 2 and 3) govern probabilities of exchanges between states:
\begin{itemize}[itemsep=0pt,topsep=0pt]
\item Algorithm 1 (Stochastic model – constant or polynomial forcing term)\\
$[e^{t\mathbf{A}}]_{ij}$ is the probability that $[\mathbf{u}_{0}]_{j}$ contributes $1/N([\mathbf{u}_{0}]_{j})$ of its value to $[\widehat{\mathbf{u}}(t)]_{i}$, and $[\widetilde{\varphi}_{\ell}(t\mathbf{A})]_{ij}$ is the probability that $\frac{t^{\ell}}{\ell}[\mathbf{b}_{\ell}(t)]_{j}$ contributes $1/N(\frac{t^{\ell}}{\ell}[\mathbf{b}_{\ell}(t)]_{j})$ of its value to $[\widehat{\mathbf{u}}(t)]_{i}$ for all $\ell = 1,\hdots,p$.
\item Algorithm 2 (Stochastic model – general time varying forcing term)\\
$[e^{\tau\mathbf{A}}]_{ij}$ is the probability that $[\widehat{\mathbf{u}}_{k}]_{j}$ contributes $1/N([\widehat{\mathbf{u}}_{k}]_{j})$ of its value to $[\widehat{\mathbf{u}}_{k+1}]_{i}$, $[\widetilde{\varphi}_{1}(\tau\mathbf{A})]_{ij}$ is the probability that $[\tau\mathbf{b}(t_{k})]_{j}$ contributes $1/N(\tau[\mathbf{b}(t_{k})]_{j})$ of its value to $[\widehat{\mathbf{u}}_{k+1}]_{i}$, and $[\widetilde{\varphi}_{2}(\tau\mathbf{A})]_{ij}$ is the probability that $[\tau(\mathbf{b}(t_{k+1})-\mathbf{b}(t_{k}))]_{j}$ contributes $1/N([\tau(\mathbf{b}(t_{k+1})-\mathbf{b}(t_{k}))]_{j})$ of its value to $[\widehat{\mathbf{u}}_{k+1}]_{i}$.
\item Algorithm 3 (Stochastic model – general state varying forcing term)\\
$[e^{\tau\mathbf{A}}]_{ij}$ is the probability that $[\widehat{\mathbf{u}}_{k}]_{j}$ contributes $1/N([\widehat{\mathbf{u}}_{k}]_{j})$ of its value to $[\widehat{\mathbf{u}}_{k+1}]_{i}$ and $[\widetilde{\varphi}_{1}(\tau\mathbf{A})]_{ij}$ is the probability that $[\tau\mathbf{b}(\widehat{\mathbf{u}}_{k})]_{j}$ contributes $1/N([\tau\mathbf{b}(\widehat{\mathbf{u}}_{k})]_{j})$ of its value to $[\widehat{\mathbf{u}}_{k+1}]_{i}$.
\end{itemize}
A key attraction of the above stochastic models is that neither time $t$ (Algorithm \ref{alg:stochastic_model1}) nor the time step $\tau$ (Algorithms \ref{alg:stochastic_model2} and \ref{alg:stochastic_model3}) need to be constrained to ensure probabilities are in $[0,1]$, unlike stochastic models based on standard explicit time discretisation methods, where the time step is considerably restricted by such probability constraints \cite{carr_2025,carwood_2026}.

\normalsize
\section{Examples}
\label{sec:examples}
We now present some illustrative examples of the deterministic and stochastic models considered in Sections \ref{sec:application}. In total, we consider seven problems, each arising from spatial discretisation of a partial differential equation model:

\begin{itemize}[itemsep=0pt,topsep=0pt]
\item Problem 1: Diffusion equation with spatially-varying source term\\ 
$c_{t} = Dc_{xx} + R(x)$ subject to $c(x,0) = f(x)$, $c_{x}(0,t) = 0$, $c_{x}(L,t) = 0$, where $D = 0.004$, $R(x) = \exp(-3000(x-0.5)^2)$ and $f(x)=0$.
\item Problem 2: Diffusion equation with spatially-varying diffusivity and polynomial boundary fluxes\\
$c_{t} = \left(D(x)c_{x}\right)_{x}$ subject to $c(x,0) = f(x)$,  $D(0)c_{x}(0,t) = p_{0}(t)$, $D(L)c_{x}(L,t) = p_{L}(t)$, where $D(x)=0.1[1-0.8(H(x-0.25)-H(x-0.75))]$, $p_{0}(t) = \alpha_{1}+\alpha_{2}t$, $p_{L}(t) = \beta_{1}+\beta_{2}t$ and $f(x) = 0$ with $\alpha_{1} = \alpha_{2} = \beta_{1} = 0.15$, $\beta_{2} = 0.1$ and $H(\cdot)$ denoting the Heaviside function.
\item Problem 3: Advection diffusion equation with non-polynomial boundary flux\\
$c_{t} = Dc_{xx} - vc_{x}$ subject to $c(x,0) = f(x)$, $Dc_{x}(0,t) - vc(0,t) = Q(t)$, $Dc_{x}(L,t) - vc(L,t) = 0$, where $D=0.01$, $v=0.7$, $f(x) = 0$ and $Q(t)=60te^{-20t}$.
\item Problem 4: Diffusion equation with Robin boundary conditions\\
$c_{t} = Dc_{xx}$ subject to $c(x,0) = f(x)$, $Dc_{x}(0,t) = \sigma(c(0,t)-C_{0})$, $Dc_{x}(L,t) = \sigma(C_{L}-c(L,t))$, where $D=0.2$, $f(x)=0$, $\sigma=2.5$, $C_{0}=1$ and $C_{L}=0.5$.
\item Problem 5: Fisher-KPP equation\\
$c_{t} = Dc_{xx} + R(c)$, subject to $c(x,0) = f(x)$, $c_{x}(0,t) = 0$, $c_{x}(L,t) = 0$, where $D=0.03$, $R(c) = 9c(1-c)$ and $f(x) = 0.1(H(x-0.45)-H(x-0.55))$.
\item Problem 6: Allen-Cahn equation\\
$c_{t} = Dc_{xx} + R(c)$, subject to $c(x,0) = f(x)$, $c_{x}(0,t) = 0$, $c_{x}(L,t) = 0$, where $D=0.001$, $R(c) = 15c(1-c)(c-0.5)$ and $f(x) = 0.1(\exp(-40(x-0.5)^{2}))+0.45$.
\item Problem 7: 2D Allen-Cahn equation\\
$c_{t} = D(c_{xx}+c_{yy}) + R(c)$, subject to $c(x,y,0) = f(x,y)$, $c_{x}(0,y,t) = 0$, $c_{x}(L,y,t) = 0$, $c_{y}(x,0,t) = 0$, $c_{y}(x,L,t) = 0$, where $D=0.001$, $R(c) = 30c(1-c)(c-0.5)$ and $f(x,y)\sim U(0,1)$.
\end{itemize}

Spatial discretisation is carried out using a standard finite volume method. For the 1D problems, the spatial domain $[0,L]$ is discretised using $n$ uniformly-spaced nodes, $x_{i} = (i-1)h$ for $i = 1,\hdots,n$ where $h = L/(n-1)>0$. The 1D finite volume surrounding node $i$ is then defined by the interval $\Omega_{i} = [x_{i}^{-},x_{i}^{+}]$ of length $V_{i} = x_{i}^{+}-x_{i}^{-}$, where $x_{i}^{-}=\max(0,x_{i}-h/2)$ and $x_{i}^{+} = \min(x_{i}+h/2,L)$. For the 2D problem, the spatial domain $[0,L]\times[0,L]$ is discretised using an $m\times m$ uniformly-spaced grid with nodes at $(x,y) = (\widehat{x}_{i},\widehat{y}_{j})$ for $i = 1,\hdots,m$ and $j = 1,\hdots,m$, where $\widehat{x}_{i} = (i-1)h$, $\widehat{y}_{j} = (j-1)h$ and $h = L/(m-1)$. The $n = m^{2}$ nodes are ordered left-to-right bottom-to-top such that node $i$ is located at position $x = x_{i} = (p(i)-1)h$ and $y = y_{i}= (q(i)-1)h$ where $p(i) = i - \lfloor\frac{i-1}{m}\rfloor m$ and $q(i) = \lfloor\frac{i-1}{m}\rfloor + 1$. The 2D finite volume surrounding node $i$ is then defined by the square region $\Omega_{i}=[x_{i}^{-},x_{i}^{+}]\times[y_{i}^{-},y_{i}^{+}]$ of area $V_{i} = (x_{i}^{+}-x_{i}^{-})(y_{i}^{+}-y_{i}^{-})$, where $x_{i}^{-} = \max(0,x_{i}-h/2)$, $x_{i}^{+} = \min(x_{i}+h/2,L)$, $y_{i}^{-} = \max(0,y_{i}-h/2)$ and $y_{i}^{+} = \min(y_{i}+h/2,L)$. 

With the $n$ nodes/states defined, we consider a solution/state vector $\mathbf{c}(t)$, where $[\mathbf{c}(t)]_{i} = c_{i} \approx c(x_{i},t)$ (Problems 1--6) or $c(x_{i},y_{i},t)$ (Problem 7). We then integrate the governing PDE over each finite volume, insert the boundary conditions and invoke standard approximations (central differences for diffusion terms, averaging for any advection terms etc). In summary, for each problem, this yields a system of differential equations in $\mathbf{c}(t)$ with initial condition $\mathbf{c}(0) = \mathbf{c}_{0} = (f(x_{1}),\hdots,f(x_{n}))^{T}$ (Problems 1--6) or $(f(x_{1},y_{1}),\hdots,f(x_{n},y_{n}))^{T}$ (Problem 7):

\begin{itemize}[itemsep=0pt,topsep=0pt]
\item Problem 1: $\mathbf{c}^{\prime}(t) = \widetilde{\mathbf{A}}\mathbf{c} + \widetilde{\mathbf{b}}$ where $\widetilde{\mathbf{b}} = (R(x_{1}),\hdots,R(x_{n}))^{T}$ and $\widetilde{\mathbf{A}}$ is a tridiagonal matrix with nonzero entries $[\widetilde{\mathbf{A}}]_{1,1} = -2D/h^{2}$ and $[\widetilde{\mathbf{A}}]_{1,2} = 2D/h^{2}$ (first row); $[\widetilde{\mathbf{A}}]_{i,i-1} = D/h^{2}$, $[\widetilde{\mathbf{A}}]_{i,i} = -2D/h^{2}$ and $[\widetilde{\mathbf{A}}]_{i,i+1} = D/h^{2}$ for $i = 2,\hdots,n-1$ (middle rows); and $[\widetilde{\mathbf{A}}]_{n,n-1} = 2D/h^{2}$, $[\widetilde{\mathbf{A}}]_{n,n} = -2D/h^{2}$ (last row). Note that $\widetilde{\mathbf{A}}$ is a row transition-rate matrix. 
\item Problem 2: $\mathbf{c}^{\prime}(t) = \widetilde{\mathbf{A}}\mathbf{c} + \widetilde{\mathbf{b}}_{1} + t\widetilde{\mathbf{b}}_{2}$ where $\widetilde{\mathbf{b}}_{\ell} = (2\alpha_{\ell}/h,0,\hdots,0,2\beta_{\ell}/h)$ for each $\ell\in\{1,2\}$ and $\widetilde{\mathbf{A}}$ is a tridiagonal matrix with non-zero entries $[\widetilde{\mathbf{A}}]_{1,1} = -2D_{1}/h^{2}$ and $[\widetilde{\mathbf{A}}]_{1,2} = 2D_{1}/h^{2}$ (first row); $[\widetilde{\mathbf{A}}]_{i,i-1} = D_{i-1}/h^{2}$, $[\widetilde{\mathbf{A}}]_{i,i} = -(D_{i-1}+D_{i})/h^{2}$ and $[\widetilde{\mathbf{A}}]_{i,i+1} = D_{i}/h^{2}$ for $i = 2,\hdots,n-1$ (middle rows); and $[\widetilde{\mathbf{A}}]_{n,n-1} = 2D_{n-1}/h^{2}$, $[\widetilde{\mathbf{A}}]_{n,n} = -2D_{n-1}/h^{2}$ (last row), where $D_{i} = D((x_{i}+x_{i+1})/2)$. Note that $\widetilde{\mathbf{A}}$ is a row transition-rate matrix. 
\item Problem 3: $\mathbf{c}^{\prime}(t) = \widetilde{\mathbf{A}}\mathbf{c} + \widetilde{\mathbf{b}}(t)$ where $\widetilde{\mathbf{b}}(t) = (2Q(t)/h,0,\hdots,0)^{T}$ and $\widetilde{\mathbf{A}}$ is a tridiagonal matrix with non-zero entries $[\widetilde{\mathbf{A}}]_{1,1} = -(2D+vh)/h^{2}$ and $[\widetilde{\mathbf{A}}]_{1,2} = (2D-vh)/h^{2}$ (first row); $[\widetilde{\mathbf{A}}]_{i,i-1} = (2D+vh)/(2h^{2})$, $[\widetilde{\mathbf{A}}]_{i,i} = -2D/h^{2}$ and $[\widetilde{\mathbf{A}}]_{i,i+1} = (2D-vh)/(2h^{2})$ for $i = 2,\hdots,n-1$ (middle rows); $[\widetilde{\mathbf{A}}]_{n,n-1} = (2D+vh)/h^{2}$, $[\widetilde{\mathbf{A}}]_{n,n} = -(2D-vh)/h^{2}$ (last row). Note that $\widetilde{\mathbf{A}}$ is not a row-transition rate matrix but is only essentially nonnegative if $0 < h \leq 2D/|v|$. 
\item Problems 4--6: $\mathbf{c}^{\prime}(t) = \widetilde{\mathbf{A}}\mathbf{c} + \widetilde{\mathbf{b}}(\mathbf{c})$ where $\widetilde{\mathbf{A}}$ is the row transition-rate matrix defined in Problem 1, and $\widetilde{\mathbf{b}}(\mathbf{c}) = (2\sigma(C_{0}-c_{1})/h,0,\hdots,0,2\sigma(C_{L}-c_{n})/h)^{T}$ (Problem 4) or $\widetilde{\mathbf{b}}(\mathbf{c}) = (R(c_{1}),\hdots,R(c_{n}))^{T}$ (Problems 5--6). 
\item Problem 7: $\mathbf{c}^{\prime}(t) = \widetilde{\mathbf{A}}\mathbf{c} + \widetilde{\mathbf{b}}(\mathbf{c})$ where $\widetilde{\mathbf{A}}$ is the usual five-point stencil matrix obtained by discretising the diffusion term with zero flux boundary conditions, where $\smash{[\widetilde{\mathbf{A}}]_{ij} = \frac{D\alpha_{i,j}}{V_{i}}}$ if $j\in S_{i}$, $[\widetilde{\mathbf{A}}]_{ij} = -\frac{D}{V_{i}}\sum_{k\in S_{i}}\alpha_{i,k}$ if $j=i$, or $[\widetilde{\mathbf{A}}]_{ij} = 0$ otherwise, where $S_{i} = \{i-1 \, |\,x_{i} > 0\}\cup\{i+1 \, |\,x_{i} < L\}\cup\{i-m \, |\,y_{i} > 0\}\cup\{i+m \, |\,y_{i} < L\}$ is the set of nodes connected to node $i$, and $\alpha_{i,j} = 1/2$ if both nodes $i$ and $j$ lie on the domain boundary (i.e. ($x_{i}\in\{0,L\}$ or $y_{i}\in\{0,L\}$) and ($x_{j}\in\{0,L\}$ or $y_{j}\in\{0,L\}$), or $\alpha_{i,j} = 1$ otherwise. Note that $\widetilde{\mathbf{A}}$ is a row transition-rate matrix since it is essentially nonnegative and has zero row sums with $\sum_{j=1}^{n}[\mathbf{A}]_{ij} = [\mathbf{A}]_{ii} + \sum_{j\in S_{i}}[\mathbf{A}]_{ij} = 0$ for all $i = 1,\hdots,n$.
\end{itemize}

Following Sections \ref{sec:ODE_systems1} and \ref{sec:ODE_systems2}, we compute solutions to the above systems using either the exact solutions $\smash{\mathbf{c}(t) = e^{t\widetilde{\mathbf{A}}}\mathbf{c}_{0} + t\varphi_{1}(t\widetilde{\mathbf{A}})\widetilde{\mathbf{b}}}$ (Problem 1) and $\smash{\mathbf{c}(t) = e^{t\widetilde{\mathbf{A}}}\mathbf{u}_{0} + t\varphi_{1}(t\widetilde{\mathbf{A}})\widetilde{\mathbf{b}}_{1} + t^{2}\varphi_{2}(t\widetilde{\mathbf{A}})\widetilde{\mathbf{b}}_{2}}$ (Problem 2) or the approximate solutions $\smash{\mathbf{c}_{k+1} = e^{\tau\widetilde{\mathbf{A}}}\mathbf{c}_{k} + \tau\varphi_{1}(\tau\widetilde{\mathbf{A}})\widetilde{\mathbf{b}}(t_{k}) + \tau\varphi_{2}(\tau\widetilde{\mathbf{A}})[\widetilde{\mathbf{b}}(t_{k+1})-\widetilde{\mathbf{b}}(t_{k})]}$ (Problem 3) and $\smash{\mathbf{c}_{k+1} = e^{\tau\widetilde{\mathbf{A}}}\mathbf{c}_{k} + \tau\varphi_{1}(\tau\widetilde{\mathbf{A}})\widetilde{\mathbf{b}}(\mathbf{c}_{k})}$ (Problems 4--7). Next, we convert each system into one involving a column transition-rate matrix by introducing the change of variables $\mathbf{c}(t) = \mathbf{V}^{-1}\mathbf{u}(t)$, where $\mathbf{V} = \text{diag}(V_{1},\hdots,V_{n})\in\mathbb{R}^{n\times n}$ and $\mathbf{V}^{-1} = \text{diag}(V_{1}^{-1},\hdots,V_{n}^{-1})$, recalling that $V_{i}$ is the length/area of the finite volume $\Omega_{i}$. For each problem, this produces a system of differential equations in $\mathbf{u}(t)$ with initial condition $\smash{\mathbf{u}(0)=\mathbf{u}_{0}=\mathbf{V}^{-1}\mathbf{c}_{0}}$ where $\smash{\mathbf{A}=\mathbf{V}\widetilde{\mathbf{A}}\mathbf{V}^{-1}}$:
\begin{itemize}[itemsep=0pt,topsep=0pt]
\item Problem 1: $\mathbf{u}^{\prime}(t) = \mathbf{A}\mathbf{u} + \mathbf{b}$ where $\mathbf{b} = \mathbf{V}\widetilde{\mathbf{b}}$ and $\mathbf{A}$ is a tridiagonal column-transition rate matrix with nonzero entries: $[\mathbf{A}]_{1,1} = -2D/h^{2}$ and $[\mathbf{A}]_{2,1} = 2D/h^{2}$ (first column); $[\mathbf{A}]_{j-1,j} = D/h^{2}$, $[\mathbf{A}]_{j,j} = -2D/h^{2}$ and $[\mathbf{A}]_{j+1,j} = D/h^{2}$ for $j = 2,\hdots,n-1$ (middle columns); $[\mathbf{A}]_{n-1,n} = 2D/h^{2}$, $[\mathbf{A}]_{n,n} = -2D/h^{2}$ (last column). 
\item Problem 2: $\mathbf{u}^{\prime}(t) = \mathbf{A}\mathbf{u} + \mathbf{b}_{1} + t\mathbf{b}_{2}$ where $\mathbf{b}_{\ell} = \mathbf{V}\widetilde{\mathbf{b}}_{\ell}$ for each $\ell\in\{1,2\}$ and $\mathbf{A}$ is a tridiagonal column-transition rate matrix with nonzero entries: $[\mathbf{A}]_{1,1} = -2D_{1}/h^{2}$ and $[\mathbf{A}]_{2,1} = 2D_{1}/h^{2}$ (first column); $[\mathbf{A}]_{j-1,j} = D_{j-1}/h^{2}$, $[\mathbf{A}]_{j,j} = -(D_{j-1}+D_{j})/h^{2}$ and $[\mathbf{A}]_{j+1,j} = D_{j}/h^{2}$ for $j = 2,\hdots,n-1$ (middle columns); $[\mathbf{A}]_{n-1,n} = 2D_{n-1}/h^{2}$, $[\mathbf{A}]_{n,n} = -2D_{n-1}/h^{2}$ (last column). 
\item Problem 3: $\mathbf{u}^{\prime}(t) = \mathbf{A}\mathbf{u} + \mathbf{b}(t)$ where $\mathbf{b}(t) = \mathbf{V}\widetilde{\mathbf{b}}(t)$ and $\mathbf{A}$ is a tridiagonal column-transition rate matrix (provided $h\leq 2D/|v|$) with nonzero entries: $[\mathbf{A}]_{1,1} = -(2D+vh)/h^{2}$ and $[\mathbf{A}]_{2,1} = (2D+vh)/h^{2}$ (first column); $[\mathbf{A}]_{j-1,j} = (2D-vh)/(2h^{2})$, $[\mathbf{A}]_{j,j} = -2D/h^{2}$ and $[\mathbf{A}]_{j+1,j} = (2D+vh)/(2h^{2})$ for $j = 2,\hdots,n-1$ (middle columns); $[\mathbf{A}]_{n-1,n} = (2D-vh)/h^{2}$, $[\mathbf{A}]_{n,n} = -(2D-vh)/h^{2}$ (last column). 
\item Problems 4--6: $\mathbf{u}^{\prime}(t) = \mathbf{A}\mathbf{u} + \mathbf{b}(\mathbf{u})$ where $\mathbf{b}(\mathbf{u}) = \mathbf{V}\widetilde{\mathbf{b}}(\mathbf{V}^{-1}\mathbf{u})$ and $\mathbf{A}$ is the column transition-rate matrix defined in Problem 1. 
\item Problem 7: $\mathbf{u}^{\prime}(t) = \mathbf{A}\mathbf{u} + \mathbf{b}(\mathbf{u})$ where $\mathbf{b}(\mathbf{u}) = \mathbf{V}\widetilde{\mathbf{b}}(\mathbf{V}^{-1}\mathbf{u})$ and $\mathbf{A}$ is a column transition-rate matrix with entries $\smash{[\mathbf{A}]_{ij} = \frac{D\alpha_{i,j}}{V_{j}}}$ if $i\in S_{j}$, $\smash{[\mathbf{A}]_{ij} = -\frac{D}{V_{j}}\sum_{k\in S_{j}}\alpha_{k,j}}$ if $i=j$, or $[\mathbf{A}]_{ij} = 0$ otherwise. Here, $\mathbf{A}$ is a column transition-rate matrix since it is essentially nonnegative and has zero column sums with $\sum_{i=1}^{n}[\mathbf{A}]_{ij} = [\mathbf{A}]_{jj} + \sum_{i\in S_{j}}[\mathbf{A}]_{ij} = 0$ for all $i = 1,\hdots,n$.
\end{itemize}

As discussed in Section \ref{sec:stochastic_models}, the above systems can be interpreted as stochastic processes consisting of individual discrete units that act independently. Here, we use the exact solutions $\mathbf{u}(t) = e^{t\mathbf{A}}\mathbf{u}_{0} + t\varphi_{1}(t\mathbf{A})\mathbf{b}_{1}$ (Problem 1) and $\mathbf{u}(t) = e^{t\mathbf{A}}\mathbf{u}_{0} + t\varphi_{1}(t\mathbf{A})\mathbf{b}_{1} + t^{2}\varphi_{2}(t\mathbf{A})\mathbf{b}_{2}$ (Problem 2) and approximate solutions $\mathbf{u}_{k+1} = e^{\tau\mathbf{A}}\mathbf{u}_{k} + \tau\varphi_{1}(\tau\mathbf{A})\mathbf{b}(t_{k}) + \tau\varphi_{2}(\tau\mathbf{A})[\mathbf{b}(t_{k+1})-\mathbf{b}(t_{k})]$ (Problem 3) and $\mathbf{u}_{k+1} = e^{\tau\mathbf{A}}\mathbf{u}_{k} + \tau\varphi_{1}(\tau\mathbf{A})\mathbf{b}(\mathbf{u}_{k})$ (Problems 4--6), together with Algorithm \ref{alg:stochastic_model1} (Problems 1--2), Algorithm \ref{alg:stochastic_model2} (Problem 3) and Algorithm \ref{alg:stochastic_model3} (Problem 4--7) to generate stochastic realisations $\widehat{\mathbf{u}}(t)$ (Algorithm 1) and $\widehat{\mathbf{u}}_{k}$ (Algorithms 2--3) of $\mathbf{u}(t)$ and $\mathbf{u}_{k}$, respectively. To compute the matrix functions of the form $\varphi_{\ell}(\gamma\mathbf{A})$ we use MATLAB's \texttt{expm} function combined with the identity provided in the proof of Proposition \ref{prop:1}(i), which remains valid when $\mathbf{A}$ is replaced with $\gamma\mathbf{A}$ as mentioned in Remark \ref{rem:1}. The stochastic realisations $\widehat{\mathbf{u}}(t)$ (Problem 1--2) and $\widehat{\mathbf{u}}_{k}$ (Problems 3--7) can then be mapped back to the original variable via the relationships $\widehat{\mathbf{c}}(t) = \mathbf{V}^{-1}\widehat{\mathbf{u}}(t)$ (Problems 1--2) or $\widehat{\mathbf{c}}_{k} = \mathbf{V}^{-1}\widehat{\mathbf{u}}_{k}$ (Problems 3--7) and compared to the corresponding deterministic solution, $\mathbf{c}(t)$ (Problems 1--2) or $\mathbf{c}_{k}$ (Problems 3--7), at specified discrete times. Results in Figures \ref{fig:plot1}--\ref{fig:plot3} compare the deterministic solution to one realisation of the stochastic model and depict the corresponding positions of the individual discrete units over time. Note that the time step of $\tau = 0.005$ used in Problems 4 and 5 exceeds the constraint $\tau \leq h^{2}/(2D)$ (i.e. $\tau \leq 0.00025$ for Problem 4 and $\tau \leq 0.0017$ for Problem 5) required for the equivalent stochastic model based on a forward Euler time discretisation \cite{carr_2025}. For each problem, the given stochastic simulation closely mimic the spatiotemporal dynamics of the underlying deterministic PDE model with numerical experimentation confirming decreased variability as $N_{s}$ is increased towards the $N_{s}\rightarrow\infty$ continuum limit.

\begin{figure}[p]
\centering
\includegraphics[width=1.0\textwidth]{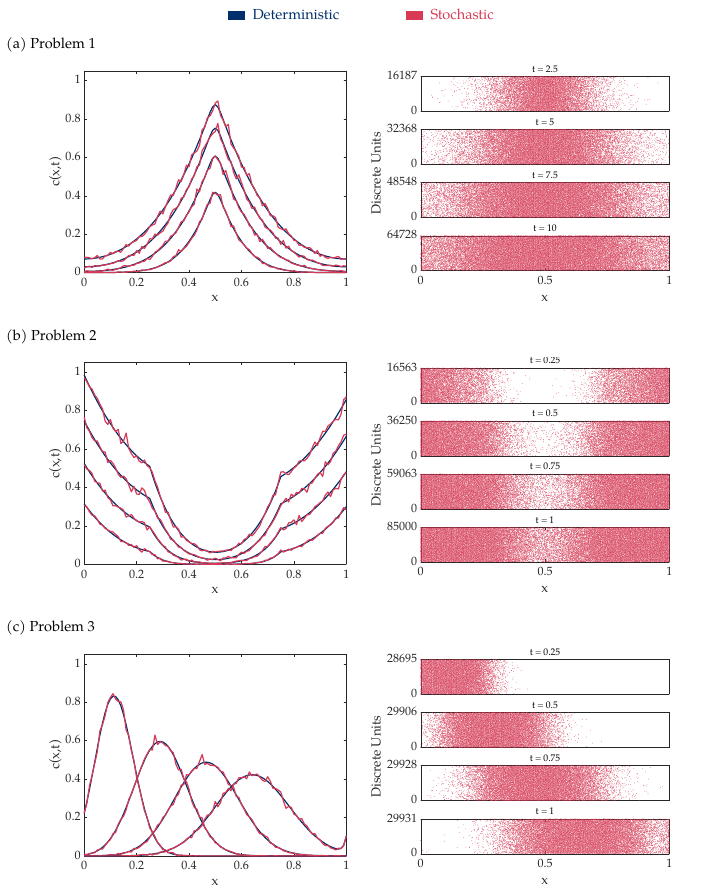}
\caption{\textbf{[Problem 1 and 2]} (Left panels) Deterministic solution [plot of $(x_{i},[\mathbf{c}(t_{k})]_{i})$] compared to one realisation of the stochastic model [plot of $(x_{i},[\widehat{\mathbf{c}}(t_{k})]_{i})$] for $i = 1,\hdots,n$ and $k = 1,\hdots,4$, where $\widehat{\mathbf{c}}(t) = \mathbf{V}^{-1}\widehat{\mathbf{u}}(t)$ with $\widehat{\mathbf{u}}(t)$ obtained from Algorithm \ref{alg:stochastic_model1}. (Right panels) Spatial distribution of discrete units where the number of discrete units in finite volume $i$ at time $t_{k}$ is given by $N([\widehat{\mathbf{u}}(t_{k})]_{i})$ and the upper limit on the vertical axes is the corresponding total number of discrete units in the system i.e. $\sum_{i=1}^{n}N([\widehat{\mathbf{u}}(t_{k})]_{i})$. \textbf{[Problem 3]} (Left panels) Deterministic solution [plot of $(x_{i},[\mathbf{c}_{k}]_{i})$] compared to one realisation of the stochastic model [plot of $(x_{i},[\widehat{\mathbf{c}}_{k}]_{i})$] for $i = 1,\hdots,n$ and $k = k_{1},k_{2},k_{3},k_{4}$, where $\widehat{\mathbf{c}}_{k} = \mathbf{V}^{-1}\widehat{\mathbf{u}}_{k}$ with $\widehat{\mathbf{u}}_{k}$ obtained from Algorithm \ref{alg:stochastic_model3}. (Right panels) Spatial distribution of discrete units where the number of discrete units in finite volume $i$ at time $t_{k}$ is given by $N([\widehat{\mathbf{u}}_{k}]_{i})$ and the upper limit on the vertical axes is the corresponding total number of discrete units in the system i.e. $\sum_{i=1}^{n}N([\widehat{\mathbf{u}}_{k}]_{i})$. Parameters: $L = 1$, $n = 101$, $h = 0.01$, $N_{s}=2\times 10^{5}$ (Problems 1--3); $[t_{1},t_{2},t_{3},t_{4}] = [2.5,5,7.5,10]$ (Problem 1), $[t_{1},t_{2},t_{3},t_{4}] = [0.25,0.5,0.75,1]$ (Problem 2); $M=100$, $\tau = 0.01$, $[k_{1},k_{2},k_{3},k_{4}] = [25,50,75,100]$ (i.e. indices corresponding to $t=0.25,0.5,0.75,1$) (Problem 3).}
\label{fig:plot1}
\end{figure}

\begin{figure}[p]
\centering
\includegraphics[width=1.0\textwidth]{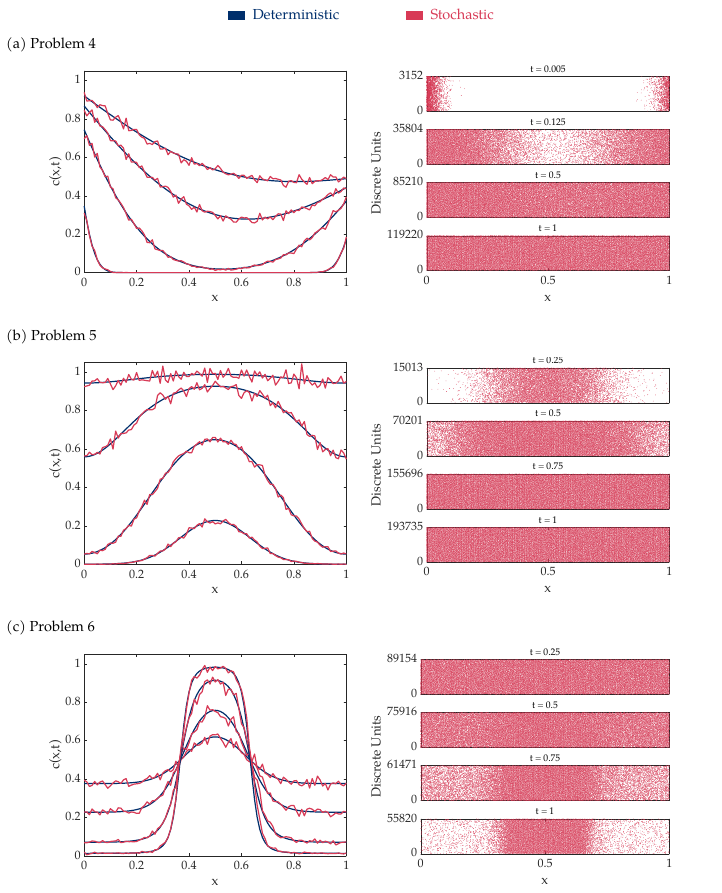}
\caption{\textbf{[Problems 4--6]} (Left panels) Deterministic solution [plot of $(x_{i},[\mathbf{c}_{k}]_{i})$] compared to one realisation of the stochastic model [plot of $(x_{i},[\widehat{\mathbf{c}}_{k}]_{i})$] for $i = 1,\hdots,n$ and $k = k_{1},k_{2},k_{3},k_{4}$, where $\widehat{\mathbf{c}}_{k} = \mathbf{V}^{-1}\widehat{\mathbf{u}}_{k}$ with $\widehat{\mathbf{u}}_{k}$ obtained from Algorithm \ref{alg:stochastic_model3}. (Right panels) Spatial distribution of discrete units where the number of discrete units located in finite volume $i$ at time $t_{k}$ is given by $N([\widehat{\mathbf{u}}_{k}]_{i})$ and the upper limit on the vertical axes is the corresponding total number of discrete units in the system i.e. $\sum_{i=1}^{n}N([\widehat{\mathbf{u}}_{k}]_{i})$. Parameters: $L = 1$, $n = 101$, $h = 0.01$, $T=1$, $M=400$, $\tau = 0.0025$, $N_{s} = 2\times 10^{5}$ (Problems 4--6); $[k_{1},k_{2},k_{3},k_{4}] = [2,50,200,400]$ (i.e. indices corresponding to $t=0.005,0.125,0.5,1$) (Problem 4); $[k_{1},k_{2},k_{3},k_{4}] = [100,200,300,400]$ (i.e. indices corresponding to $t=0.25,0.5,0.75,1$) (Problems~5~and~6).}
\label{fig:plot2}
\end{figure}

\begin{figure}[p]
\centering
\includegraphics[width=1.0\textwidth]{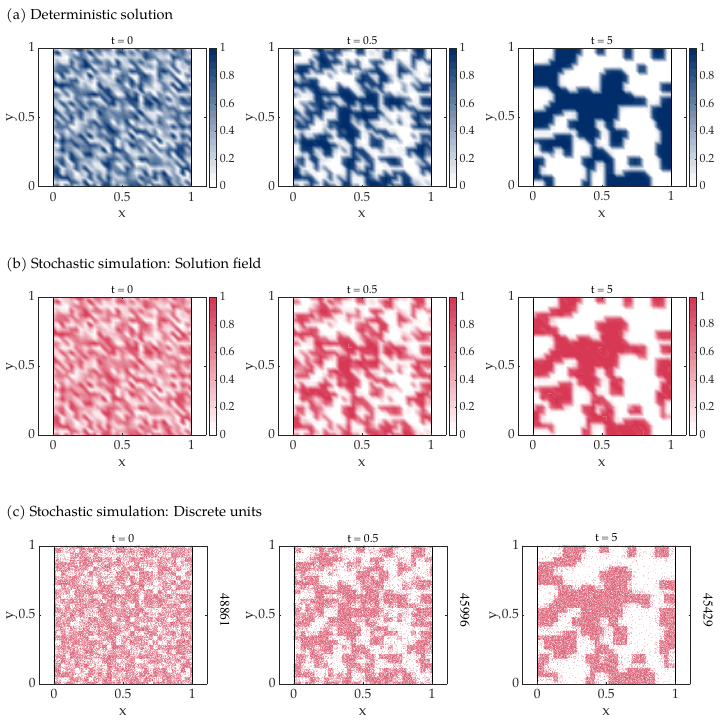}
\caption{\textbf{[Problem 7]} (a) Deterministic solution [contour plot of $(x_{i},y_{i},[\mathbf{c}_{k}]_{i})$ for $i = 1,\hdots,n$ and $k = k_{1},k_{2},k_{3}$]. (b) One realisation of the stochastic model [contour plot of $(x_{i},y_{i},[\widehat{\mathbf{c}}_{k}]_{i})$ for $i = 1,\hdots,n$ and $k = k_{1},k_{2},k_{3}$], where $\widehat{\mathbf{c}}_{k} = \mathbf{V}^{-1}\widehat{\mathbf{u}}_{k}$ with $\widehat{\mathbf{u}}_{k}$ obtained from Algorithm \ref{alg:stochastic_model3}. (c) Spatial distribution of discrete units where the number of discrete units located in finite volume $i$ at time $t_{k}$ is given by $N([\widehat{\mathbf{u}}_{k}]_{i})$ and the number on the right of the plot provides the corresponding total number of discrete units in the system  i.e. $\sum_{i=1}^{n}N([\widehat{\mathbf{u}}_{k}]_{i})$. Parameters: $L = 1$, $m = 31$, $n = 31^{2}$, $h = 1/30$, $T=5$, $M=200$, $\tau = 0.025$, $N_{s} = 10^{5}$, $[k_{1},k_{2},k_{3}] = [0,20,200]$ (i.e. indices corresponding to $t=0,0.5,5$).}
\label{fig:plot3}
\end{figure}

\section{Conclusion}
A fundamental property of the matrix exponential is that $e^{\mathbf{A}}$ is a row (column) stochastic matrix if $\mathbf{A}$ is a row (column) transition-rate matrix. In this paper, we have studied the generalisation of this property to the matrix phi-functions, $\varphi_{\ell}(\mathbf{A}) = \sum_{k=0}^{\infty}\mathbf{A}^{k}/(k+\ell)!$ where $\ell$ is a positive integer, which arise in exact and approximate solutions of certain systems of ordinary differential equations. Specifically, the main results/outcomes of the paper include
\begin{enumerate}[(i),itemsep=0pt,topsep=0pt]
\item Properties satisfied by the matrix phi-functions including $\widetilde{\varphi}_{\ell}(\mathbf{A}) = \ell!\varphi_{\ell}(\mathbf{A})$ is a row (column) stochastic matrix if $\mathbf{A}$ is a row (column) transition-rate matrix.
\item Insights into the nonnegative and conservative nature of exact and approximate solutions of some linear and semi-linear systems of ordinary differential equations with constant, time-dependent and state-dependent forcing/source terms.
\item New stochastic models and corresponding algorithms that can be used to generate stochastic simulations of the ODE systems in (ii) by using the stochastic properties in (i) and by assuming the continuous states can be partitioned into individual discrete units that act independently.
\end{enumerate}
The new stochastic models generalise previous work concerning linear ODE systems without a forcing term and as demonstrated in Section \ref{sec:examples} can be used to generate stochastic simulations that mimic the solution behaviour of some common deterministic partial differential equations. A key attraction of such stochastic models is that they can quantify uncertainty and variability in model predictions by carrying out repeated simulations of the stochastic algorithm and constructing probability distributions enveloping the deterministic solution (one for each spatial node). Such probability distributions can then be used to ask statistical inference questions about model predictions.

\subsection*{Data availability} 
Supporting MATLAB code to be made available on GitHub once paper is published.

\newpage
{\begingroup
\setlength{\bibsep}{1.5pt}
\selectfont\small
\bibliographystyle{unsrt}
\bibliography{references}
\endgroup}

\end{document}